\documentclass[10pt,journal]{IEEEtran}
\usepackage{subfigure}
\usepackage{cite}
\usepackage{graphicx}
\usepackage{psfrag}
\usepackage{url}
\usepackage{color}

\usepackage{amsmath}
\usepackage{array}
\usepackage{graphics}
\usepackage{epsfig}
\usepackage{amsbsy}
\usepackage{amssymb}
\usepackage{amsthm}
\usepackage{amsmath}
\newtheorem{theorem}{Theorem}
\newtheorem{lemma}[theorem]{Lemma}

\newtheorem{remark}[theorem]{Remark}

\usepackage{framed}
\usepackage{algorithm}
\usepackage{algorithmic}
\usepackage{mathtools}
\usepackage{epstopdf}

\newtheorem{proposition}{Proposition} 

\IEEEoverridecommandlockouts

\begin{document}
\title{Resource Allocation and Fairness in Wireless Powered Cooperative Cognitive Radio Networks}

\author{Sanket S. Kalamkar,~\IEEEmembership{Student~Member,~IEEE,}
Jeya Pradha Jeyaraj,\\
Adrish Banerjee,~\IEEEmembership{Senior~Member,~IEEE,} 
and~Ketan Rajawat,~\IEEEmembership{Member,~IEEE} 
\thanks{S. S. Kalamkar, A. Banerjee, and K. Rajawat are with the Department of Electrical Engineering, Indian Institute of Technology Kanpur, 208016, India (e-mail: \{kalamkar, adrish, ketan\}@iitk.ac.in).}
\thanks{J. P. Jeyaraj is with Singapore University of Technology and Design, 487372, Singapore (e-mail: jeyapradhaj@gmail.com).}
\thanks{S. S. Kalamkar is supported by the Tata Consultancy Services (TCS) research fellowship.}

}

\maketitle

\begin{abstract}
We integrate a wireless powered communication network with a cooperative cognitive radio network, where multiple secondary users (SUs) powered wirelessly by a hybrid access point (HAP) help a primary user relay the data. As a reward for the cooperation, the secondary network gains the spectrum access where SUs transmit to HAP using time division multiple access. To maximize the sum-throughput of SUs, we present a secondary sum-throughput optimal resource allocation (STORA) scheme. Under the constraint of meeting target primary rate, the STORA scheme chooses the optimal set of relaying SUs and jointly performs the time and energy allocation for SUs. Specifically, by exploiting the structure of the optimal solution, we find the order in which SUs are prioritized to relay primary data. Since the STORA scheme focuses on the sum-throughput, it becomes inconsiderate towards individual SU throughput, resulting in low fairness. To enhance fairness, we investigate three resource allocation schemes, which are (i) equal time allocation, (ii) minimum throughput maximization, and (iii) proportional time allocation. Simulation results reveal the trade-off between sum-throughput and fairness. The minimum throughput maximization scheme is the fairest one as each SU gets the same throughput, but yields the least SU sum-throughput.
 \end{abstract}
\section{Introduction}
In the design of future wireless networks, energy harvesting has gained importance due to its ability to furnish energy to wireless devices while bestowing the freedom of mobility. Harvesting energy from ambient sources like solar, wind, and thermoelectric effects is a self-sustaining and green approach to prolong the lifetime of wireless devices~\cite{paradiso}. But, the uncertain and intermittent energy arrivals can make these sources unreliable for wireless applications with strict quality-of-service (QoS) requirements. This limitation has motivated the use of dedicated radio-frequency (RF) signals radiated by an access point in a controlled manner to power wireless devices, giving rise to wireless powered communication networks (WPCNs)~\cite{ju,suzhi,tabassum}. The RF power transfer is the backbone of WPCN, where a hybrid access point (HAP)\footnote{The term \textit{hybrid} comes from the fact that HAP can act as an energy access point (EAP) and a data access point (DAP)~\cite{zhou}. However, EAP and DAP may be located separately.} broadcasts energy to devices via RF signals. The devices use the harvested energy to transmit their data back to HAP.

The WPCN offers multiple advantages over harvesting energy from ambient RF signals that are not intended for energy transfer, such as co-channel interference. First, the energy broadcast by HAP is controllable in terms of system design (e.g., number of antennas), waveform design, and transmit power level. This allows HAP to supply stable and continuous energy to wireless devices~\cite{tabassum,zhou}. Contrary, ambient RF signals are not controllable as other RF sources may not be always transmitting. The second advantage is the flexibility in deployment of HAP~\cite{kaibin, zhang_deployment}. That is, HAP can be placed at a convenient location to reduce losses in energy transfer due to path-loss and shadowing.

Currently, using dedicated energy broadcast, it is possible to transfer energy in order of microwatts over a distance of 10 meters~\cite{powercast}, which is sufficient for several low-power wireless nodes~\cite{lowpower} such as medical implants, wireless sensors, and RFID. With the recent technological advancements like low-power electronics~\cite{lowpower}, massive MIMO (tens to hundreds of antennas) providing sharp energy beams towards users to improve energy transfer efficiency~\cite{zhao_massive}, and small-cell deployment reducing inter-nodal distances~\cite{andrew_femto}, WPCN seems to have the potential to meet the energy requirements of wireless devices with stringent QoS demands.

The WPCN has received attention in different setups~\cite{naderi,suzhi,tabassum,ju,ju4,kaibin,zhou,zhang_deployment,seung,kang,usercoop,chen3,zhong}. In~\cite{ju, ju4,kang}, the proposed harvest-and-transmit policy maximizes the sum-throughput in WPCN, where multiple users harvest energy from the energy broadcast by HAP and use the harvested energy to transmit the information to HAP in a time division multiple access (TDMA) manner. The works in~\cite{kaibin} and \cite{naderi} study WPCN in a random-access-network and cellular network, respectively. Reference~\cite{seung} examines a new type of WPCN, where wireless charging vehicles furnish wireless power to mobile nodes. The works in~\cite{usercoop} and \cite{chen3} manifest the user cooperation in WPCN, where the user nearer to HAP relays the data of the distant user to maximize the weighted sum-rate. Authors in~\cite{zhong} calculate the achievable throughput of a cooperative relay-assisted communication, where a power beacon wirelessly powers the source and the relay.

Given the dramatic increase in different wireless applications in the limited spectrum, WPCN will have to coexist and share the spectrum with other established communication networks.
Thus, WPCN, if made \textit{cognitive}~\cite{gridlock}, can efficiently share the spectrum with an existing network without degrading the latter's QoS. One such way is the cooperation between WPCN and the existing communication network. For example, in a typical cognitive radio network, the cooperation between primary users (PUs) and secondary users (SUs) improves the spectral efficiency of the network~\cite{gridlock}. In such cooperative cognitive radio network (CCRN), SUs relay PU data to gain the spectrum access for their own transmissions as a reward. This cooperation enhances PU's QoS resulting in reduced PU transmission time, which provides SUs an opportunity to transmit their own data in the remaining time. Also, consider a case when the primary direct link is in deep fade and PU fails to meet its target rate. Then, the data cooperation from SUs provides user diversity, which improves the chances of meeting the target primary rate. For these reasons, the cooperation is a win-win strategy for both PU and SUs. Now, acting as a secondary network, WPCN can assist the existing network (primary network) to relay its data and achieve spectrum access as a reward without requiring extra spectrum.

A number of works have studied CCRN in a non-energy harvesting setup~\cite{leasing,krik2,pandhari:2010,hossain,active,pandhari:2009, manna,cao1, hao1,long}. For multiuser CCRN, the work in~\cite{long} proposes the best SU selection for relaying to maximize the SU sum-throughput. But, in an energy harvesting scenario, the energy factor becomes crucial, and even the best SU to relay PU data may fail to meet the primary rate constraint due to insufficient energy. This energy constraint may often lead to the selection of multiple SUs for relaying. As to wireless energy harvesting in CCRN, in~\cite{infoenergy,wang1,quang,zhai}, an SU harvests RF energy from a PU's transmission which it uses to relay PU data and transmit its own data. In~\cite{sixing}, under the save-then-transmit protocol, authors develop an optimal cooperation strategy for an SU that harvests energy from ambient radio signals. In~\cite{jeya1}, PU supplies energy to an SU to facilitate the data cooperation between PU and SU. But, the works in~\cite{infoenergy,wang1,quang,zhai,sixing,jeya1} consider a single-SU CCRN and assume that an SU harvests energy from either PU transmissions or ambient radio signals.

A recent work~\cite{zhang:CRPN} investigates a cognitive wireless powered network, where it assumes non-causal knowledge of PU data. Contrary, we consider that SUs have to spend some time to receive and decode PU data before relaying.
Also, \cite{zhang:CRPN} considers that a conventionally powered HAP from the secondary network (not energy harvesting SUs) transmits PU data; while in our proposed protocol, since multiple \textit{energy harvesting} SUs may relay PU data, we provide the optimal SU selection strategy to decide which SUs will relay PU data, which is one of our main contributions. These fundamental differences in the system model lead to an entirely different problem formulation than the one in \cite{zhang:CRPN}. Moreover, we address the fairness issue for the wireless powered CCRN.

\begin{figure}
\centering
\includegraphics[scale=0.14]{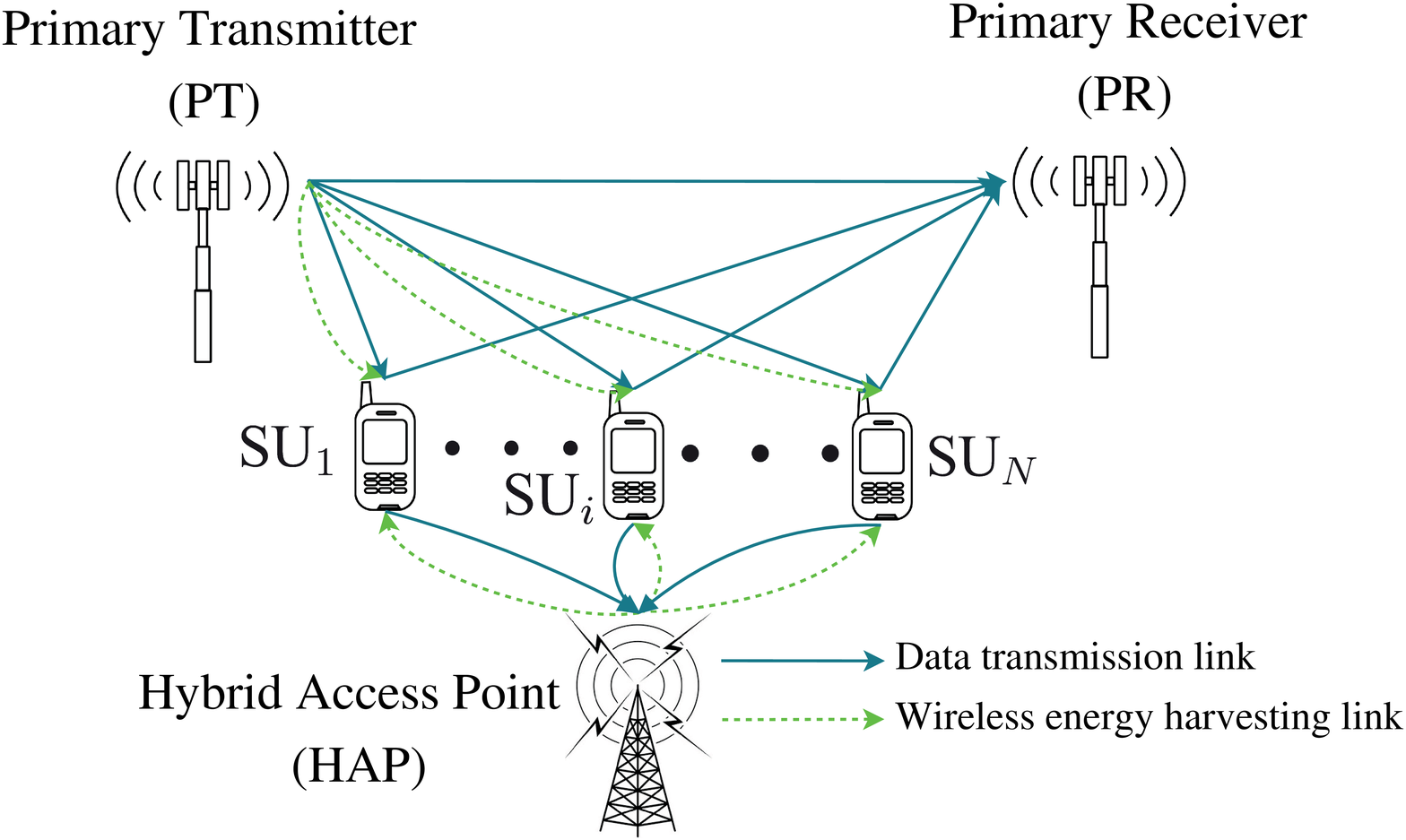}
\caption{Wireless powered cooperative cognitive radio network.}
\label{fig:sys}
\end{figure}

\subsection{Contributions}
As Fig.~\ref{fig:sys} shows, we integrate WPCN with CCRN to develop wireless powered CCRN (WP-CCRN) and seize the benefits of WPCN and CCRN together. In WP-CCRN, multiple SUs of the secondary network harvest energy from the RF broadcast by HAP, while the primary network shares its spectrum with the secondary network if the cooperation from the latter can meet its target rate. 
We now briefly discuss the contributions and main results of the paper below.

1) \textit{Cooperation protocol}: We propose a cooperation protocol to facilitate the spectrum sharing between a PU pair and multiple wireless powered SUs. The cooperation protocol will allow SUs to harvest energy, receive PU data, relay PU data, and transmit their own data. 

2) \textit{Optimal solution}: For the proposed protocol, we present an SU sum-throughput optimal resource allocation (STORA) scheme that maximizes SU sum-throughput under primary rate constraint. The STORA scheme chooses the optimal set of relaying SUs from the set of SUs that could decode PU data successfully. The original optimization problem is non-convex due to the unknown set of SUs that can successfully decode PU data and the product of optimization variables. But, for a fixed decoding set and with change of variables, the problem reduces to a convex problem. Then, we provide an iterative algorithm to obtain the global optimal solution of the original non-convex problem.

3) \textit{Resource allocation}: Based on the analysis of the optimal solution, we answer the following three questions related to the allocation of aforementioned resources: (i) Which SUs will relay PU data? (ii) At the relaying SUs, how to split the harvested energy between PU data relaying and own transmissions? (iii) How to divide the rewarded access time among SUs? Interestingly, the order in which SUs are prioritized to relay PU data is independent of the energy harvested by each SU. We further show that the optimal resource allocation from the secondary network point of view is inconsiderate of the individual SU throughput, resulting in low fairness.

4) \textit{Fairness}: To enhance fairness, we consider the following three resource allocation schemes: (i) \textit{equal time allocation} (ETA), (ii) \textit{minimum throughput maximization} (MTM), and (iii) \textit{proportional time allocation} (PTA).
Through simulation results, we compare STORA scheme with three aforementioned fairness enhancing schemes and reveal the trade-off between sum-throughput and fairness. The MTM scheme is the fairest; however, the secondary network achieves the least SU sum-throughput. On the other hand, the STORA scheme is the most unfair scheme, but achieves the highest SU sum-throughput. 

\subsection{Organization of the Paper}
In Section~\ref{sec:sys}, we describe the system model and the cooperation protocol between primary and secondary networks under different system constraints. In Section~\ref{sec:STORA}, we propose the STORA scheme and find the underlying principle behind the selection of SUs to relay PU data. Then, in Section~\ref{sec:fair_schemes}, we study three fairness enhancing resource allocation schemes that consider different fairness criteria. Section~\ref{sec:results} presents simulation results and compares the sum-throughput and fairness performance for all four resource allocation schemes against different system parameters. In Section~\ref{sec:pos_var}, we discuss some possible variations to our proposed cooperation protocol. Finally, we make concluding remarks in Section~\ref{sec:conclusion}.

\textit{Notation}: Table~\ref{tab:notation} lists the notations used in the paper. Unless otherwise stated, transmit power of a user, energy, noise power, and throughput are given in Watts, Joules, Watts/Hz, and Nats/s/Hz, respectively.
\begin{table}
\caption{List of Notations} \label{tab:notation}
\begin{center}
\renewcommand{\arraystretch}{1.05}
\begin{tabular}{c  p{7cm} }
\hline 
 {\bf Notation} & {\hspace{2cm}}{\bf Definition}
\\
\hline
\hline 

$N$ & Number of secondary users \\
SU$_i$ & $i$th secondary user \\
$h_{\mathrm{p}}$ & Channel power gain between primary transmitter and primary receiver\\
$h_{\mathrm{p}i}$ & Channel power gain between primary transmitter and $i$th secondary user \\
$h_{i\mathrm{p}}$ & Channel power gain between $i$th secondary user and primary receiver\\
$h_{\mathrm{h}i}$ & Channel power gain between hybrid access point and $i$th secondary user \\
$h_{i\mathrm{h}}$ & Channel power gain between $i$th secondary user and hybrid access point \\
$T$ & One fading block duration \\
$t_{\mathrm{e}}$ & Secondary users' energy harvesting duration \\
$t_{\mathrm{0}}$ & Secondary users' primary data reception as well as relaying duration \\
$t_\mathrm{a}$ & Total rewarded access time to secondary users \\
$t_i$ & Access time for $i$th secondary user in the rewarded period \\
$P_{\mathrm{p}}$ & Primary transmit power\\
$P_{\mathrm{e}}$ & Hybrid access point's transmit power\\
$P_{i\mathrm{h}}$ & $i$th secondary user's transmit power on its direct link to hybrid access point\\
$\boldsymbol{P_{\mathrm{sh}}}$ & A vector of length $N$ given by $[P_{1\mathrm{h}}, \dotsc, P_{N\mathrm{h}}]$\\
$P_{i\mathrm{p}}$ & $i$th secondary user's transmit power on its relaying link to primary receiver\\
$\boldsymbol{P_{\mathrm{sp}}}$ & A vector of length $N$ given by $[P_{1\mathrm{p}}, \dotsc, P_{N\mathrm{p}}]$\\
$N_0$ & Noise power\\
$E_{i}$ & Energy harvested by $i$th secondary user \\
$E_{i\mathrm{h}}$ & Energy spent by $i$th secondary user on its direct link to hybrid access point\\
$\boldsymbol{E_{\mathrm{sh}}}$ & A vector of length $N$ given by $[E_{1\mathrm{h}}, \dotsc, E_{N\mathrm{h}}]$\\
$E_{i\mathrm{p}}$ & Energy spent by $i$th secondary user on its relaying link to primary receiver\\
$\boldsymbol{E_{\mathrm{sp}}}$ & A vector of length $N$ given by $[E_{1\mathrm{p}}, \dotsc, E_{N\mathrm{p}}]$\\
$\eta$ & Energy harvesting efficiency factor ($0 < \eta \leq 1$)\\
$\bar{R}_{\mathrm{p}}$ & Target primary rate\\
$R_{\mathrm{p,c}}$ & Primary rate achieved under cooperation\\
$R_{i}$ & Rate achieved by $i$th secondary user\\
$\mathcal{S_D}$ & Decoding set \\
$|\mathcal{S}|$ & Cardinality of set $\mathcal{S}$ \\
$\mathcal{W}(\cdot)$ & Lambert W function \\
$x^{*}$ & Optimal value of a variable $x$\\
\hline 
\end{tabular}
\end{center}
\end{table}
\section{System Model and Cooperation Protocol}
\label{sec:sys}
Fig.~\ref{fig:sys} shows the system model, where a primary pair consisting of a primary transmitter (PT) and a primary receiver (PR) coexists with $N$ secondary users (SUs) that communicate with a hybrid access point (HAP). All nodes have a single antenna. The HAP has a stable and conventional energy supply, while SUs harvest energy from RF energy broadcast by HAP. At a time, an SU can receive either data or energy, but not both. Also, HAP cannot broadcast energy and receive data from SUs simultaneously. An $i$th SU is denoted by SU$_{i}$, $i \in \lbrace{ 1,\dotsc, N \rbrace } $. Let $h_{\mathrm{p}}$, $h_{\mathrm{p}i}$, $h_{i\mathrm{p}}$, $h_{\mathrm{h}i}$, and $h_{i\mathrm{h}}$ denote channel power gains of PT-PR, PT-SU$_{i}$, SU$_{i}$-PR, HAP-SU$_{i}$, and SU$_{i}$-HAP, respectively.

We focus on a scenario where PU fails to meet its target rate via direct link and thus seeks cooperation from the secondary network.\footnote{The other possible cooperation scenario is the one, where PU alone can meet its target rate, but can achieve an even better rate with secondary cooperation.} Under this cooperation, SUs use the energy harvested from HAP's RF broadcast to relay PU data to PR which augments PU direct link transmission. This improves the received signal-to-noise ratio (SNR) at PR, which in turn enhances the primary rate. If PU meets its target rate through secondary cooperation, SUs may use the remaining duration of the time-slot to transmit their own data to HAP. The primary rate constraint requires SUs to (i) allot time to receive PU data and relay it which reduces the time for energy harvesting and spectrum
access, (ii) consume energy to relay PU data which reduces the energy available for the spectrum access; both reducing SU sum-throughput. Thus, the primary rate constraint has a doubly negative effect on the SU sum-throughput.

Before we describe the cooperation protocol, we list below the assumptions considered in the protocol.

\begin{itemize}
\item[A1.] The PR can coherently combine the data received from the primary direct link and the relaying links from SUs using maximal ratio combining (MRC)~\cite{hossain,long,jeya1,manna,infoenergy,cao1,pandhari:2010,active}.
\item[A2.] An SU relays PU data to PR in a decode-and-forward manner~\cite{hao1,long,pandhari:2010,pandhari:2009}.
\item[A3.] In the rewarded period, SUs transmit to HAP using time division multiple access (TDMA) due to its simplicity~\cite{ju,ju4,kang,hossain,long}.
\item[A4.] All channels are independent and experience quasi-static fading, where the channel gains remain constant during one block of the transmission and change independently from one block to another. We assume the perfect knowledge of all channel power gains~\cite{cao1,long,infoenergy, jeya1,manna,krik2,zhang:CRPN}.\footnote{Due to cooperation between PU and SUs, we assume that the knowledge about channels corresponding to PU can be obtained from the PU itself~\cite{jovicic}. For the purpose of exposition, we skip the details of the process of obtaining the knowledge about channels corresponding to PU.}$^{,}$\footnote{The SU sum-throughput obtained with perfect channel state information (CSI) acts as an upper bound for the sum-throughput obtained in the case of imperfect CSI. The study of the effect of resource expenditure to acquire CSI is out of scope of this paper.}
\item[A5.] HAP can act as a central controller to facilitate the cooperation between primary and secondary networks. It also coordinates among users, selects SUs to relay PU data, and maintains synchronization among secondary nodes for relaying PU data and transmitting their own data.
\end{itemize}
Considering that the SUs harvest energy, receive PU data, relay PU data, and access the channel, we propose a cooperation protocol that divides a fading block of duration $T$ into four phases as depicted in Fig.~\ref{fig:protocol}. The final phase of SUs' spectrum access is subdivided into multiple slots where SUs transmit to HAP using TDMA. The following subsection discusses this four-phase cooperation protocol between PU and SUs.
\begin{figure}
\centering
\includegraphics[scale=0.42]{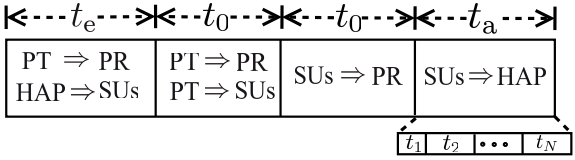}
\caption{Cooperation protocol between primary and secondary users.}
\label{fig:protocol}
\end{figure}
 
\subsection{Cooperation Protocol}
\label{sec:coop_prot}
\subsubsection{Phase One}
HAP broadcasts \textit{deterministic} RF signals with maximum power $P_{\mathrm{e}}$ Watts, allowed by the peak power constraint, from which SUs harvest energy. Simultaneously, PT broadcasts data with power $P_{\mathrm{p}}$ Watts. The SUs can use PT's transmission as a source of extra energy. Since HAP's energy signal does not carry any information, it can be made deterministic, and thus PR can be made aware of it before the start of the transmission~\cite{zhang_deter,chen_deter}. In this way, PR can remove the interference caused by concurrent HAP's broadcast. Thus, in phase one, we can write PU's achievable throughput as
\begin{equation}
 R_{\mathrm{p,1}} = \frac{t_\mathrm{e}}{T} \ln \left(1 + \frac{h_{\mathrm{p}}P_{\mathrm{p}}}{\Gamma N_0}\right),
\label{eq:phase1_rate}
\end{equation}
where $N_0$ is the additive white Gaussian noise (AWGN) power in Watts/Hz at PR. {Also, $\Gamma$ denotes the SNR gap from the AWGN channel capacity due to practical modulation and coding schemes used~\cite{goldsmith_book}.} Meanwhile, SU$_i$ harvests energy
\begin{equation}
E_{i} = \eta \big(P_{\mathrm{e}} h_{\mathrm{h}i} + P_{\mathrm{p}}h_{\mathrm{p}i}\big)t_{\mathrm{e}},
\label{eq:harv_en}
\end{equation}
where $\eta$ is the energy conversion efficiency factor with $0 < \eta \leq 1$. The term $P_{\mathrm{e}} h_{\mathrm{h}i}$ in~\eqref{eq:harv_en} corresponds to the energy harvested from HAP's broadcast, while the term $P_{\mathrm{p}}h_{\mathrm{p}i}$ corresponds to the energy harvested from PT's transmission. The SUs use the energy harvested in this phase to relay PU data and transmit their own data in phases three and four, respectively.

\subsubsection{Phases Two and Three}
{Similar to the cooperation protocol in~\cite{pandhari:2010,pandhari:2009,infoenergy}, to receive and relay PU data, the cooperation from SUs happens over two phases of equal duration, i.e., over phases two and three.} In phase two of duration $t_0$, PU broadcasts its residual data after the phase one transmission, which SUs attempt to decode. The SUs that could decode the complete PU data successfully form the decoding set $\mathcal{S_D}$. An SU$_i$ decodes the data successfully if no outage happens in phase two~\cite{laneman,bletsas}, i.e., the rate on PT-SU$_i$ link is able to support the incoming data from PT.
In this phase, the received SNR at PR due to primary direct link transmission is given by $\frac{h_{\mathrm{p}}P_{\mathrm{p}}}{\Gamma N_{0}}$.

In phase three, if the decoding set $\mathcal{S_D}$ is non-empty, then SUs chosen from it relay PU data.
Before decoding the received signal, PR employs MRC. Thus, the resultant SNR at PR becomes the sum of SNRs over two channel uses, one on the primary direct link in phase two and the other on the secondary relaying link in phase three. Consequently, the achievable primary throughput with secondary cooperation becomes
\begin{equation}
R_{\mathrm{p,2}} = \frac{t_0}{T}\ln\left(1+\frac{h_{\mathrm{p}}P_{\mathrm{p}}}{\Gamma N_{0}}  + \frac{\sum_{\mathrm{SU}_i \in \mathcal{S_D}}h_{i\mathrm{p}}P_{i\mathrm{p}} }{\Gamma N_{0}}\right),
\label{eq:rel_r}
\end{equation}
where $P_{i\mathrm{p}}$ is the power spent by SU$_{i}$ on its relaying link. The term $\frac{\sum_{\mathrm{SU}_i \in \mathcal{S_D}}h_{i\mathrm{p}}P_{i\mathrm{p}} }{\Gamma N_{0}}$ in~\eqref{eq:rel_r} corresponds to the SNR contribution from the relaying links of SUs that could decode PU data successfully. The power $P_{i\mathrm{p}} = 0$ for SU$_{i} \in \mathcal{S_D}$ that does not relay PU data as well as for SUs that could not decode PU data.

Given the transmissions on the primary direct link in phases one and two, and on the cooperative secondary relaying link in phase three, we can now write the overall primary throughput achieved under primary-secondary cooperation as
\begin{align}
R_{\mathrm{p,c}} & = R_{\mathrm{p,1}} + R_{\mathrm{p,2}}  = \frac{t_{\mathrm{e}}}{T}\ln \left(1 + \frac{h_{\mathrm{p}}P_{\mathrm{p}}}{\Gamma N_0}\right) \nonumber \\
& + \frac{t_0}{T}\ln\left(1+\frac{h_{\mathrm{p}}P_{\mathrm{p}}}{\Gamma N_{0}} + \frac{\sum_{\mathrm{SU}_i \in \mathcal{S_D}}h_{i\mathrm{p}}P_{i\mathrm{p}} }{\Gamma N_{0}}\right).
\label{eq:coopr}
\end{align}

\subsubsection{Phase Four}

The final phase is the rewarded period to SUs of duration $t_{\mathrm{a}}$, where SUs may access the channel and transmit their own information to HAP in a TDMA fashion. Let $t_i$ denote the time allocated to SU$_i$ and $t_{\mathrm{a}} = \sum_{i=1}^{N} t_i$. Then, the rate achieved by SU$_i$ is given by
\begin{equation}
R_{i} = \frac{t_{i}}{T}\ln\left(1+ \frac{h_{i\mathrm{h}}P_{i\mathrm{h}} }{\Gamma N_{0}} \right),
\label{eq:Rsi}
\end{equation}
where $P_{i\mathrm{h}}$ is the power used by SU$_i$ to transmit its data.

\subsection{System Constraints}
\label{sec:sys_con}

We now define the system constraints that are common to four resource allocation schemes considered in Sections~\ref{sec:STORA} and \ref{sec:fair_schemes}. 
\subsubsection{Primary Rate Constraint}
The PU's achievable rate $R_{\mathrm{p,c}}$ under cooperation must meet the target primary rate $\bar{R}_{\mathrm{p}}$. That is,
\begin{equation}
R_{\mathrm{p,c}} \geq \bar{R}_{\mathrm{p}}.
\label{eq:pqos}
\end{equation} 
In other words, the cooperation benefits PU if $\bar{R}_{\mathrm{p}}T$ amount of its data is delivered to PR successfully through it.
Using~\eqref{eq:coopr}, we can write the primary rate constraint as
\begin{align}
& \bigg[{t_\mathrm{e}} \ln \left(1 + \frac{h_{\mathrm{p}}P_{\mathrm{p}}}{ \Gamma N_0}\right) \nonumber \\
& + t_0\ln\left(1+\frac{h_{\mathrm{p}}P_{\mathrm{p}}}{\Gamma N_{0}} + \frac{\sum_{\mathrm{SU}_i \in \mathcal{S_D}}h_{i\mathrm{p}}P_{i\mathrm{p}} }{\Gamma N_{0}}\right)\bigg]\geq \bar{R}_{\mathrm{p}}T.
\label{eq:pqos1}
\end{align}

\subsubsection{Decoding Constraint at SUs}

In phase three, SUs that have decoded PU data in phase two can only participate in relaying.  An SU$_i$ can decode PU data successfully, i.e., SU$_i \in \mathcal{S_D}$ if
\begin{equation}
t_0 \ln \left(1 + \frac{h_{\mathrm{p}i}P_{\mathrm{p}}}{\Gamma N_{0}} \right) \geq \bar{R}_{\mathrm{p}}T - t_{\mathrm{e}} \ln \left(1 + \frac{h_{\mathrm{p}}P_{\mathrm{p}}}{\Gamma N_0}\right),
\label{eq:dec_con}
\end{equation}
where the term $t_0 \ln \left(1 + \frac{h_{\mathrm{p}i}P_{\mathrm{p}}}{\Gamma N_{0}} \right)$ is the amount of data from PT that PT-SU$_i$ link can support, $\bar{R}_{\mathrm{p}}T$ is the total amount of data PU wishes to send to PR, and $R_{\mathrm{p,1}}T = t_{\mathrm{e}} \ln \left(1 + \frac{h_{\mathrm{p}}P_{\mathrm{p}}}{\Gamma N_0}\right)$ denotes the amount of data that PT could send successfully to PR in phase one (see \eqref{eq:phase1_rate}). Thus, the right hand side term of \eqref{eq:dec_con} represents the residual data after PT's direct link transmission in phase one. Therefore, the decoding constraint~\eqref{eq:dec_con} implies that, to satisfy the primary rate constraint given in \eqref{eq:pqos1}, an SU should be able to decode the residual PU data after phase one successfully.

\subsubsection{Energy Neutrality Constraint at SUs}
The amount of energy spent by SU$_{i}$ on the relaying and access links cannot exceed the amount of its harvested energy from HAP's energy broadcast and PT's transmission in phase one. That is,
\begin{equation}
P_{i\mathrm{p}}t_{0} + P_{i\mathrm{h}}t_{i} \leq \eta \big(P_{\mathrm{e}} h_{\mathrm{h}i} + P_{\mathrm{p}}h_{\mathrm{p}i}\big)t_{\mathrm{e}}, \quad i = 1, 2, \dotsc, N,
\label{eq:enc}
\end{equation}
where the terms $P_{i\mathrm{p}}t_{0}$ and $P_{i\mathrm{h}}t_{i}$ denote the energy spent by SU$_i$ on its relaying and access links, respectively.

\subsubsection{Total Time Constraint}
The total time spent by SUs on energy harvesting in phase one of duration $t_{\mathrm{e}}$, PU data reception and decoding in phase two of duration $t_0$, relaying PU data in phase three of duration $t_0$, and accessing the channel for their own transmission in final phase of duration $ T - t_{\mathrm{e}} - 2t_0 = \sum_{i=1}^{N} t_i$, cannot exceed the slot duration $T$. That is, 
\begin{equation}
t_{\mathrm{e}}+2t_{0}+\sum_{i=1}^{N}t_{i} \leq T.
\label{eq:tnc}
\end{equation}
Given the constraints in~\eqref{eq:pqos1}-\eqref{eq:tnc}, our goal is to optimally allocate the time for each phase of the cooperation protocol, choose SUs for relaying, and divide the harvested energy at the relaying SUs between PU data relaying and own data transmission to HAP.{\footnote{ In this work, we do not include the delay due to decode and forward relaying and the time needed to reconfigure the transceiver.}} Hereafter, we assume $T = 1$ without loss of generality.

\section{Sum-Throughput Optimal Resource Allocation}
\label{sec:STORA}

In this section, our objective is to choose the optimal set of relaying SUs and find the optimal time and energy allocation for SUs that maximize the SU sum-throughput\footnote{An alternate objective could be weighted sum-throughput maximization.} given by
\begin{equation}
R_{\mathrm{s,sum}} = \sum_{i=1}^{N} R_i= \sum_{i=1}^{N}t_{i}\ln\left(1+ \frac{h_{i\mathrm{h}}P_{i\mathrm{h}} }{\Gamma N_{0}} \right).
\end{equation}
Given the constraints discussed in Section~\ref{sec:sys_con}, we formulate the SU sum-throughput maximization problem as follows:
\begin{align}
\mathop{\mathrm{maximize}}_{\mathcal{S_{D}}, \boldsymbol{P_{\mathrm{s}\mathrm{h}}},\boldsymbol{P_{\mathrm{s}\mathrm{p}}, \boldsymbol{t}}} & ~~~ R_{\mathrm{s,sum}} \nonumber \\ 
\mathrm {subject \ to} & ~~~  (\ref{eq:pqos1}), \eqref{eq:dec_con}, (\ref{eq:enc}), (\ref{eq:tnc}),\nonumber \\
& ~~~ t_{i} , t_{0}, t_{\mathrm{e}} \geq 0, ~~~ \forall  i, \nonumber \\
& ~~~ P_{i\mathrm{p}}, P_{i\mathrm{h}}  \geq 0, ~~~ \forall  i,
\label{prob}
\end{align}
where $\boldsymbol{P_{\mathrm{s}\mathrm{h}}} = [P_{1\mathrm{h}}, \dotsc, P_{N\mathrm{h}}]$, $\boldsymbol{P_{\mathrm{s}\mathrm{p}}} = [P_{1\mathrm{p}}, \dotsc, P_{N\mathrm{p}}]$, and $\boldsymbol{t} = [t_{\mathrm{e}}, t_{0}, \boldsymbol{t_{\mathrm{a}}}]$ with $\boldsymbol{t_{\mathrm{a}}} = [t_{1}, t_{2}, \dotsc, t_{N}]$. 

\noindent \underline{\textbf{Claim 1}}: The optimization problem \eqref{prob} is non-convex.
\begin{proof}
The energy neutrality constraint \eqref{eq:enc} is non-convex due to the product terms of optimization variables $P_{i\mathrm{p}}t_0$ and $P_{i\mathrm{h}}t_i$. Also, only SUs that belong to the decoding set $\mathcal{S_D}$ can relay PU data. But, the decoding constraint~\eqref{eq:dec_con} shows that the optimization variables $t_0$ and $t_{\mathrm{e}}$ decide which SUs can decode PU data successfully; thus, we do not know $\mathcal{S_D}$ beforehand. Unknown $\mathcal{S_{D}}$ and non-convexity of \eqref{eq:enc} make the problem~\eqref{prob} non-convex.
\end{proof}
\noindent \underline{\textbf{Claim 2}}: With the change of variables and for a fixed decoding set $\mathcal{S_D}$, the problem~\eqref{prob} reduces to a convex one.
\begin{proof}
The product of optimization variables---power and time---makes the constraint~\eqref{eq:enc} non-convex. Thus, we rewrite the power-time product as energy and normalize it by their respective $\eta h_{\mathrm{h}i}$.\footnote{Normalization is for notational simplicity.} Regarding this, we denote $\frac{P_{i\mathrm{h}}t_{i}}{\eta h_{\mathrm{h}i} } = E_{i\mathrm{h}}$, $ \frac{P_{i\mathrm{p}}t_{0}}{\eta h_{\mathrm{h}i} } =  E_{i\mathrm{p}}$, and $\frac{P_{\mathrm{p}} h_{\mathrm{p}i}}{h_{\mathrm{h}i}} = \theta_i$, and rewrite \eqref{eq:enc} as
\begin{equation}
E_{i\mathrm{p}} + E_{i\mathrm{h}} \leq (P_{\mathrm{e}} + \theta_i) t_{\mathrm{e}},  \hspace{3mm} \forall i,
\end{equation}
which is now an affine constraint. Still, the unknown $\mathcal{S_D}$ makes problem~\eqref{prob} non-convex. But, after fixing $\mathcal{S_D}$, we can reformulate~\eqref{prob} as an energy and time allocation problem as follows:
\begin{subequations}
\begin{align}
\mathop{\mathrm{maximize}}_{\boldsymbol{E_{\mathrm{s}\mathrm{h}}},\boldsymbol{E_{\mathrm{s}\mathrm{p}}}, \boldsymbol{t}} & ~ \sum_{i=1}^N t_{i}\ln\left(1+ \frac{\gamma_{i\mathrm{h}}E_{i\mathrm{h}} }{t_{i}} \right)  \label{eq:obj} \\ 
\mathrm {subject \ to} & ~  \!\bigg[Q_1 t_\mathrm{e}   + t_0\!\ln \! \left(\!\!1+\gamma_{\mathrm{p}} + \frac{\sum_{\mathrm{SU}_i \in \mathcal{S_D}}\gamma_{i\mathrm{p}}E_{i\mathrm{p}} }{t_{0}}\!\right)\!\!\bigg] \!\geq \!  \bar{R}_{\mathrm{p}}, \label{eq:qos}\\
& ~ \!\!t_0 \ln \!  \left(\!\!1 + \frac{h_{\mathrm{p}i}P_{\mathrm{p}}}{\Gamma N_{0}} \!\right) \geq \bar{R}_{\mathrm{p}} - Q_1 t_{\mathrm{e}} , \forall~  \mathrm{SU}_i \in \mathcal{S_{D}}, \label{eq:dec_con2}  \\
& ~ E_{i\mathrm{p}} + E_{i\mathrm{h}} \leq (P_{\mathrm{e}} + \theta_i) t_{\mathrm{e}} , \hspace{3mm} \forall i,  \label{eq:conv1}\\
& ~ t_{\mathrm{e}}+2t_{0}+\sum_{i=1}^{N}t_{i} \leq 1, \label{eq:conv2} \\
 & ~  t_{i}, t_{0}, t_{\mathrm{e}} \geq 0, ~~~ \forall i,  \label{eq:tpos}\\
 & ~ E_{i\mathrm{p}}, E_{i\mathrm{h}} \geq 0, ~~~ \forall i,  \label{eq:last}
\end{align}
\label{prob1}
\end{subequations}\vspace*{-5mm}

\noindent where $\boldsymbol{E_{\mathrm{s}\mathrm{h}}} = [E_{1\mathrm{h}}, \dotsc, E_{N\mathrm{h}}]$, $\boldsymbol{E_{\mathrm{s}\mathrm{p}}} = [E_{1\mathrm{p}}, \dotsc, E_{N\mathrm{p}}]$, $\gamma_{\mathrm{p}} = \frac{h_{\mathrm{p}}P_{\mathrm{p}}}{\Gamma N_{0}}$, $\gamma_{i\mathrm{p}} = \eta h_{\mathrm{h}i}\frac{h_{i\mathrm{p}}}{\Gamma N_{0}}$, $\gamma_{i\mathrm{h}} = \eta h_{\mathrm{h}i} \frac{h_{i\mathrm{h}}}{\Gamma N_{0}}$, and
\begin{equation}
Q_1 = \ln \left(1 + \frac{h_{\mathrm{p}}P_{\mathrm{p}}}{\Gamma N_0}\right).
\label{eq:Q1}
\end{equation}
As function $f(x) = \ln(1+x)$ is concave, its perspective function $f(x,y) = y\ln(1+\frac{x}{y})$ is also concave in terms of both $x$ and $y$, for all $x, y > 0$. Since the sum of concave functions is concave, we can see that the objective is concave in nature. Also, the constraint~\eqref{eq:qos} is concave and the constraints~\eqref{eq:dec_con2}-\eqref{eq:conv2} are affine. Thus, the problem~\eqref{prob1} is convex.
\end{proof}
Although the joint optimization problem in \eqref{prob} is non-convex, the problem structure is utilized to obtain a globally optimal solution. The following subsection describes the algorithm to achieve the globally optimal solution of \eqref{prob}.

\subsection{Global Optimal Solution of \eqref{prob}}
\label{sec:opt_dec_set}
Let us first sort $\boldsymbol{h_{\mathrm{ps}}} = [h_{\mathrm{p}1},\dotsc,h_{\mathrm{p}N}]$ in the decreasing order. If SU$_i$ with the channel power gain $h_{\mathrm{p}i}$ satisfies the decoding constraint~\eqref{eq:dec_con2}, SU$_j$ with $h_{\mathrm{p}j} > h_{\mathrm{p}i}$ also satisfies~\eqref{eq:dec_con2}. Let $ \boldsymbol{h}^{j}_{\boldsymbol{\mathrm{ps}}}$ denote the $j$th element of the sorted $\boldsymbol{h_{\mathrm{ps}}} $. Now, we first assume that all SUs can successfully decode PU data, i.e., $|\mathcal{S_D}| =  N$, where $|\mathcal{S_D}|$ is the cardinality of set $\mathcal{S_D}$. Then, we can rewrite the decoding constraint in~\eqref{eq:dec_con2} as
\begin{equation}
t_0 \ln \left(1 + \frac{\boldsymbol{h}^{|\mathcal{S_D}|}_{\boldsymbol{\mathrm{ps}}}P_{\mathrm{p}}}{\Gamma N_{0}} \right) \geq \bar{R}_{\mathrm{p}} - Q_1 t_{\mathrm{e}},
\label{eq:dec_con3}
\end{equation}
where $Q_1$ is given by \eqref{eq:Q1}. We then solve the optimization problem \eqref{prob1} and obtain the corresponding objective value $R_{\mathrm{s,sum}}^{|\mathcal{S_D}|}$ for $|\mathcal{S_D}| = N$ using the iterative algorithm discussed in Section~\ref{sec:opt1}. If the problem \eqref{prob1} is infeasible, set $R_{\mathrm{s,sum}}^{|\mathcal{S_D}|} = 0$. Then, we exclude $|\mathcal{S_D}|$-th element of the sorted $\boldsymbol{{h}_{\mathrm{ps}}}$ and solve the optimization problem \eqref{prob1} taking the decoding constraint \eqref{eq:dec_con3} into account. We obtain the corresponding objective value $R_{\mathrm{s,sum}}^{|\mathcal{S_D}|}$ with $|\mathcal{S_D}| = N-1$. This process is repeated until $|\mathcal{S_D}| = 1$. Finally, the optimal SU sum-throughput is given by
\begin{equation}
R_{\mathrm{s,sum}}^{*} = \mathrm{max}\left(R_{\mathrm{s,sum}}^{1}, R_{\mathrm{s,sum}}^{2}, \dotsc, R_{\mathrm{s,sum}}^{N}\right),
\end{equation}
and the corresponding $\mathcal{S_D}$ is the optimal decoding set. {The sorting of $\boldsymbol{h_{\mathrm{ps}}}$ has reduced the search for the optimal decoding set from $2^{N}-1$ (excluding null set) possibilities to $N$. This is because, we can ignore those choices of the decoding sets of SUs that contain an SU$_j$ with $h_{\mathrm{p}j} > h_{\mathrm{p}i}$ that could not decode PU data, but SU$_i$ with $h_{\mathrm{p}i}$ could decode PU data.} Algorithm~\ref{main_alg} summarizes the process of finding the global optimal solution of \eqref{prob}. 
\begin{algorithm}
\linespread{0.85}
\caption{Finding the global optimal solution of \eqref{prob}}
\label{main_alg}
\begin{itemize}
\item[1.] Set $k = N$. Fix $|\mathcal{S_D}| = k$.
\item[2.] Replace \eqref{eq:dec_con2} with \eqref{eq:dec_con3}.
\item[3.] Solve \eqref{prob1} using Algorithm~\ref{alg1} and compute the value of its objective $R_{\mathrm{s,sum}}^{|\mathcal{S_D}|}$. If \eqref{prob1} is infeasible, set $R_{\mathrm{s,sum}}^{|\mathcal{S_D}|}= 0$.
\item[4.] Set $k = k-1$. If $k < 1$, stop; else, go to step 2.
\item[5.] $R_{\mathrm{s,sum}}^{*} = \mathrm{max}\left(R_{\mathrm{s,sum}}^{1}, R_{\mathrm{s,sum}}^{2}, \dotsc, R_{\mathrm{s,sum}}^{N}\right)$.
\end{itemize}
\end{algorithm}

{In terms of complexity, the bottleneck of Algorithm~\ref{main_alg} is $N$ times application of Algorithm~\ref{alg1} having complexity $\mathcal{O}(N)$ (See complexity calculation for Algorithm 2 in Section~\ref{sec:opt1}.). Thus, Algorithm~\ref{main_alg} has complexity $\mathcal{O}(N^2)$.}

The following subsection presents an algorithm to find the optimal solution of~\eqref{prob1} for a given $\mathcal{S_D}$. The constraint~\eqref{eq:qos} shows that the energy spent on the relaying link of SU$_i$ depends on the energy spent by other users SU$_j$, $j \neq i$, on their relaying links. Also, the optimization variables $(\boldsymbol{E_{\mathrm{s}\mathrm{h}}}, \boldsymbol{E_{\mathrm{sp}}})$ depend on the optimization variables $(\boldsymbol{t_{\mathrm{a}}}, t_{0})$ as seen from \eqref{eq:obj} and \eqref{eq:qos}. Thus, the problem \eqref{prob1} is a convex problem with coupled variables and coupled constraints. This emphasizes the need for an iterative algorithm to compute the optimal solution efficiently. Though there are many standard algorithms to find the optimal solution of a convex optimization problem, given the coupled nature of problem~\eqref{prob1} we use the following hierarchical decomposition based on block coordinate descent method~\cite[Chapter 1]{bertesekas}. It allows us to break the problem~\eqref{prob1} into two convex subproblems of decoupled nature.

\subsection{Optimal Solution of \eqref{prob1} for a given $\mathcal{S_D}$}
\label{sec:opt1}
For notational simplicity, let us denote the rate of SU$_i$ as $R_{i}(E_{i\mathrm{h}},t_{i})$, i.e., $R_{i}$ as a function of $E_{i\mathrm{h}}$ and  $t_{i}$, and the rate of primary user as $R_{\mathrm{p,c}}(\boldsymbol{E_{\mathrm{sp}}} , t_{\mathrm{e}}, t_{0})$, i.e., $R_{\mathrm{p,c}}$ as a function of $\boldsymbol{E_{\mathrm{sp}}}$, $t_{\mathrm{e}}$, and $t_{0}$.

First, by fixing $t_{\mathrm{e}}$, we decompose the problem~\eqref{prob1} into two subproblems: $\mathtt{SP1}-$for solving energy allocation $(\boldsymbol{E_{\mathrm{s}\mathrm{h}}},\boldsymbol{E_{\mathrm{sp}}})$ for a fixed $(\boldsymbol{t_{\mathrm{a}}}, t_{0})$ and $\mathtt{SP2}-$for solving time allocation $(\boldsymbol{t_{\mathrm{a}}},t_{0})$ for a fixed $(\boldsymbol{E_{\mathrm{s}\mathrm{h}}},\boldsymbol{E_{\mathrm{sp}}})$. The subproblems are given as
\begin{align}
\mathtt{SP1}: \mathop{\mathrm{maximize}}_{\boldsymbol{E_{\mathrm{s}\mathrm{h}}},\boldsymbol{E_{\mathrm{s}\mathrm{p}}}} &~~~ \sum_{i=1}^N t_{i}\ln\left(1+ \frac{\gamma_{i\mathrm{h}}E_{i\mathrm{h}} }{t_{i}} \right)  \nonumber \\
\mathrm{subject~to}&~~~ \eqref{eq:qos}, ~\eqref{eq:conv1}, ~\eqref{eq:last},
\label{eq:sp1}
\end{align}
\begin{align}
\mathtt{SP2}: \mathop{\mathrm{maximize}}_{t_{0}, \boldsymbol{t_{\mathrm{a}}}} &~~~ \sum_{i=1}^N t_{i}\ln\left(1+ \frac{\gamma_{i\mathrm{h}}E_{i\mathrm{h}} }{t_{i}} \right)  \nonumber \\
\mathrm{subject~to}&~~~ ~\eqref{eq:qos}, ~\eqref{eq:dec_con2},~\eqref{eq:conv2},~\eqref{eq:tpos}.
\label{eq:sp2}
\end{align}
The subproblems are convex and they satisfy Slater's constraint qualification~\cite[Chapter 3]{bertesekas}. Hence, the Karush-Kuhn-Tucker (KKT) conditions are necessary and sufficient to find the optimal solution. 

The subproblems $\mathtt{SP1}$ and $\mathtt{SP2}$ form Level 1 of the iterative algorithm, and at Level 2, we solve the master problem to compute $t_{\mathrm{e}}$, which is also convex. The levels 1 and 2 are executed recursively until the optimization variables $(\boldsymbol{E_{{\mathrm{s}\mathrm{h}}}}, \boldsymbol{E_{{\mathrm{s}\mathrm{p}}}},  {t_{0}}, \boldsymbol{t_{\mathrm{a}}})$ converge to a predetermined accuracy~\cite{yu}, resulting in the optimal solution.

\textbf{Level 1}: The Lagrangian $\mathcal{L}_{1}$ of $\mathtt{SP1}$ is
\begin{align}
\mathcal{L}_{1} & =   \sum_{i=1}^N R_{i}(E_{i\mathrm{h}},t_{i} ) - \lambda \left(\bar{R}_{\mathrm{p}} - R_{\mathrm{p,c}}\left(\boldsymbol{E_{\mathrm{s}\mathrm{p}}}, t_{\mathrm{e}},  t_{0}\right) \right) 
\nonumber \\ 
& - \sum_{i = 1}^{N} \mu_{i}(E_{i\mathrm{p}} + E_{i\mathrm{h}} - (P_{\mathrm{e}} + \theta_i)t_{\mathrm{e}}), 
\label{eq:lag1}
\end{align}
where $\lambda$ and $\boldsymbol{\mu} = [\mu_1, \dotsc, \mu_N]$ denote the dual variables associated with the constraints \eqref{eq:qos} and \eqref{eq:conv1}, respectively. The dual problem of $\mathtt{SP1}$ is given by
\begin{equation}
\underset{\boldsymbol{\mu}}{\mathop{\mathrm{min}}} \hspace{1mm} \underset{\boldsymbol{E_{{\mathrm{s}\mathrm{h}}}}, \boldsymbol{E_{{\mathrm{s}\mathrm{p}}}}}{\mathop{\mathrm{max}}} \hspace{1mm} \mathcal{L}_{1},
\label{eq:dual_prob1}
\end{equation}
which is solved as follows. First, using KKT stationarity conditions,  we solve the primal variable $E_{i\mathrm{p}}$ along with  $E_{i\mathrm{h}}$ keeping $(E_{j\mathrm{p}}, t_{0}, \boldsymbol{t_{\mathrm{a}}})$ and $(\lambda, \boldsymbol{\mu})$ fixed. Second, we compute the dual variable $\mu_{i}$ using bisection method.\footnote{Since the constraint \eqref{eq:qos} is present in both $\mathtt{SP1}$ and $\mathtt{SP2}$, the associated dual variable $\lambda$ is solved at Level 2.}
We then repeat this two step process until $(E_{i\mathrm{h}}, E_{i\mathrm{p}})$ converge. In this manner, 
the dual problem \eqref{eq:dual_prob1} is solved for every SU$_i$ ($i \in \lbrace
1, \dotsc, N\rbrace$). Now, by associating dual variables $\kappa$ and $\nu$ with constraints~\eqref{eq:dec_con2} and \eqref{eq:conv2}, respectively, we write the Lagrangian $\mathcal{L}_{2}$ of time allocation problem $\mathtt{SP2}$ as
\begin{align}
\hspace{-2mm} \mathcal{L}_{2} & =   \sum_{i=1}^N R_{i}(E_{i\mathrm{h}},t_{i} ) - \lambda \left(\bar{R}_{\mathrm{p}} - R_{\mathrm{p,c}}\left(\boldsymbol{E_{\mathrm{s}\mathrm{p}}}, t_{\mathrm{e}},  t_{0}\right) \right) 
\nonumber \\ 
\hspace{-3mm}  & - \kappa \left( \bar{R}_{\mathrm{p}} - Q_1 t_{\mathrm{e}} - Q_2 t_{0} \right)  -  \nu\left(t_{\mathrm{e}}+2t_{0}+\sum_{i=1}^{N}t_{i}-1\right),
\label{eq:lag2}
\end{align}
where $Q_1$ is given by~\eqref{eq:Q1} and $Q_{2} = \ln \left(1 + \frac{\boldsymbol{h}^{|\mathcal{S_D}|}_{\boldsymbol{\mathrm{ps}}}P_{\mathrm{p}}}{\Gamma N_{0}} \right)$. The corresponding dual problem is given by
\begin{equation}
\underset{\lambda, \kappa, \nu}{\mathop{\mathrm{min}}} \hspace{1mm} \underset{t_{0}, \boldsymbol{t_{\mathrm{a}}}}{\mathop{\mathrm{max}}} \hspace{1mm} \mathcal{L}_{2}.
\label{eq:dual_prob2}
\end{equation}
For fixed $(\boldsymbol{E_{{\mathrm{s}\mathrm{h}}}}, \boldsymbol{E_{{\mathrm{s}\mathrm{p}}}})$ and $(\lambda, \kappa, \nu)$, the primal variables $(t_{0}, \boldsymbol{t_{\mathrm{a}}})$ can be found using KKT stationarity conditions. The dual variables that minimize $\underset{t_{0}, \boldsymbol{t_{\mathrm{a}}}}{\mathop{\mathrm{max}}} \hspace{1mm} \mathcal{L}_{2}$ are found using their gradients given by
\begin{align}
g_{\lambda} & =   R_{\mathrm{p,c}}(\boldsymbol{E}_{\boldsymbol{\mathrm{s}\mathrm{p}}}, t_{\mathrm{e}}, t_{0}) - \bar{R}_{\mathrm{p}}, \label{eq:grad_lam} \\
g_{\kappa} & =   Q_1 t_{\mathrm{e}} + Q_2 t_{0} - \bar{R}_{\mathrm{p}}, \label{eq:grad_kappa} \\
g_{\nu}  & =   1 - t_{\mathrm{e}} - 2t_{0}  - \sum_{i=1}^{N}t_{i}. \label{eq:grad_nu} 
\end{align}

\textbf{Level 2}: Let $l$ denote the iteration index of the algorithm. Having found $(\boldsymbol{E}^{l}_{\boldsymbol{{\mathrm{s}\mathrm{h}}}}, \boldsymbol{E}^{l}_{\boldsymbol{{\mathrm{s}\mathrm{p}}}}, t^{l}_{0}, \boldsymbol{t^{l}_{\mathrm{a}}})$ and $(\lambda^{l}, \kappa^{l}, \boldsymbol{\mu}^{l}, \nu^{l})$ for given $t^{l-1}_{\mathrm{e}}$ in Level $1$, we solve the master primal problem for $t^{l}_{\mathrm{e}}$ given as
\begin{equation}
t^{l}_{\mathrm{e}} = \mathop{\mathrm{arg}} \underset{t_{\mathrm{e}} }{\mathop{\mathrm{max}}} \hspace{1mm}{ \mathcal{G}(t_{\mathrm{e}})},
\label{eq:master}
\end{equation}
where $\mathcal{G}(t_{\mathrm{e}}) = \mathcal{L}_{2}|_{(t^{l}_{0}, \boldsymbol{t^{l}_{\mathrm{a}}},\lambda^{l}, \kappa^{l}, \nu^{l} )} + \sum_{i = 1}^{N} \mu^{l}_{i}(E^{l}_{i\mathrm{p}} + E^{l}_{i\mathrm{h}} - (P_{\mathrm{e}} + \theta_i)t_{\mathrm{e}})$, which is the Lagrangian of the problem \eqref{prob1} for given $t_{\mathrm{e}}$ and $\mathcal{S_D}$. Since $\mathcal{G}(t_{\mathrm{e}})$ is non-differentiable at $t^{l}_{\mathrm{e}}$, we compute $t^{l}_{\mathrm{e}}$ using its subgradient given by
\begin{equation}
g_{t_{e}} =  (\lambda^l + \kappa^l) Q_1 + \sum_{i=1}^{N} \mu^{l}_{i} (P_{\mathrm{e}} + \theta_i) - \nu^{l}.
\label{eq:grad_te}
\end{equation}
The analytical expressions of optimization variables obtained through KKT conditions are presented in Proposition \ref{lemma1} as follows.
\begin{proposition}
\label{lemma1}
The optimal solution for the optimization problem \eqref{prob1} is given by
\begin{align}
& E^{*}_{i\mathrm{h}} =  t^{*}_{i}\left[ \frac{1}{\mu^{*}_{i}} - \frac{1}{\gamma_{i\mathrm{h}}}\right]^{+},  \label{eq:Eb} \\
& E^{*}_{i\mathrm{p}}  = t^{*}_{0}\left[\frac{\lambda^{*}}{\mu^{*}_{i}} - \frac{1+\gamma_{\mathrm{p}}}{\gamma_{i\mathrm{p}}} - \frac{\sum_{{\mathrm{SU}_j \in \mathcal{S_D}}, j\neq i}\gamma_{j\mathrm{p}}E^{*}_{j\mathrm{p}}}{\gamma_{i\mathrm{p}}}\right]^{+}, \label{eq:Ep} \\
& t^{*}_{\mathrm{e}} = 1-2t^{*}_{0}-\sum_{i=1}^{N}t^{*}_{i}, \label{eq:te} \\
& t^{*}_{0} = \frac{\sum_{\mathrm{SU}_i \in \mathcal{S_D}}{\gamma_{i\mathrm{p}}}E^{*}_{i\mathrm{p}}}{\exp\left({ \mathcal{W}\left( \frac{-1}{\frac{2\nu^{*} - \kappa^* Q_2}{\lambda^{*}} + 1}  \right) + \frac{2\nu^{*}- \kappa^* Q_2}{\lambda^{*}} +1 }\right) -1 - \gamma_{\mathrm{p}}}, \label{eq:t0} \\ 
& t^{*}_{i} = \frac{{\gamma_{i\mathrm{h}}}E^{*}_{i\mathrm{h}}}{\exp\left( \mathcal{W}\left( \frac{-1}{\nu^{*} + 1}  \right) + \nu^{*} +1\right)-1}, \label{eq:ta}
\end{align}
where $[x]^{+}  \buildrel \Delta \over = \max(x,0)$ and $\mathcal{W}(\cdot)$ is the Lambert W function~\cite{corless}. 
\end{proposition}
\begin{proof}
See Appendix \ref{app:opt_soln_STORA}.
\end{proof}

Algorithm~\ref{alg1} summarizes the above iterative procedure. 

\begin{algorithm}
\linespread{0.9}
\small{
\caption{Optimal solution of~\eqref{prob1} for a given $\mathcal{S_D}$}
\label{alg1}
\begin{itemize}
\item[1.] Initialize primal variables $\boldsymbol{E}_{\boldsymbol{\mathrm{s}\mathrm{p}} }, \boldsymbol{t_{\mathrm{a}}}, t_{0} , t_{\mathrm{e}} > 0$, and dual variables $\lambda, \kappa, \nu $ $ > 0$ 
\item[2.] Set $\mu_{\min}, \mu_{\max} > 0$, $\epsilon = 10^{-5}$, step size $ \alpha > 0$, and iteration index $l = 0$ 
\item[3.] \textbf{Repeat}: $l \leftarrow l+1$ , $\alpha \leftarrow {{\alpha}/{\sqrt{l}}}^{*}$
\item[4.] \textit{Level 1}: \textbf{Repeat}
\begin{itemize}
\item[-] \texttt{To solve for (SP1)} -
\item[a.] \textbf{For} $i = 1:N$
\begin{itemize}
\item[•] \textbf{While} $|\mu_{\max} - \mu_{\min}| > \epsilon$
\begin{itemize}
\item[•]  $\mu_{i} \leftarrow (\mu_{\min} + \mu_{\max})/2$ 
\item[•] Compute $(E_{i\mathrm{h}}, E_{i\mathrm{p}})$  using $(\boldsymbol{t_{\mathrm{a}}}, t_{0})$ in \eqref{eq:Eb}-\eqref{eq:Ep} 
\item[•] If $(E_{i\mathrm{h}} + E_{i\mathrm{p}}) < (P_{\mathrm{e}} + \theta_i) t^{l-1}_{\mathrm{e}} $,  $\mu_{\max} \leftarrow \mu_{i}$; Else, $\mu_{\min} \leftarrow \mu_{i}$.
\end{itemize}
\item[•] \textbf{End}
\end{itemize} 
\item[c.] \textbf{End}
\item[-] \texttt{To solve for (SP2)} -
\item[d.] Compute $(t_{i}, t_{0})$ $ \forall i$ using $( \boldsymbol{E}_{\boldsymbol{\mathrm{s}\mathrm{h}}},  \boldsymbol{E}_{\boldsymbol{\mathrm{s}\mathrm{p}}})$ in \eqref{eq:t0}-\eqref{eq:ta}. 
\item[e.] Update the dual variables $\lambda, \kappa, \nu$ with their gradients given in $\eqref{eq:grad_lam} - \eqref{eq:grad_nu}$ using interior point method  \end{itemize}
\item[5.] \textbf{Until} convergence
\item[6.] \textit{Level 2}: Update the primal variable: $t^{l}_{\mathrm{e}} = \left[ t^{l-1}_{\mathrm{e}} + \alpha \nabla t_{\mathrm{e}} \right]^{+}$
\item[7.] \textbf{Until} all the optimization parameters converge
\item[8.] Compute the optimal powers $P^{*}_{i\mathrm{h}} = \eta h_{\mathrm{h}i} \frac{E^{*}_{i\mathrm{h}}}{t^{*}_{i}}$, $P^{*}_{i\mathrm{p}} = \eta h_{\mathrm{h}i} \frac{E^{*}_{i\mathrm{p}}}{t^{*}_{0}}$ $\forall i$ 
\end{itemize}}
$^{*}$ The step-size $\alpha$ is chosen such that it satisfies the diminishing step-size rule~\cite[Chapter 1]{bertesekas}.
\end{algorithm}

{We now calculate the complexity of Algorithm 2. The ``while loop" that uses bisection method runs for $N$ times, thus having complexity $\mathcal{O}(N)$ as $\mu_{\min}$, $\mu_{\max}$, and $\epsilon$ are constant. The other bottleneck in the loop corresponding to $l$ is the computation of $t_i$ and $t_0$ having complexity $\mathcal{O}(N)$ (step 4.d.). Since steps $4-7$ are repeated for $l$ times, the overall complexity of steps $4-7$ is $\mathcal{O}(lN)$. Fig.~\ref{fig:compl} shows the impact of average number of iterations $l$ till convergence on the algorithm for the given simulation parameters considered in Section~\ref{sec:results}, where the number of iterations used till convergence remains almost the same with increase in $N$, and thus has complexity $\mathcal{O}(1)$. Finally, $\mathcal{O}(N)$ operations are done at step 8. Thus, Algorithm 2 has complexity $\mathcal{O}(N)$.}
\begin{figure}
\centering
   \includegraphics[scale=0.38]{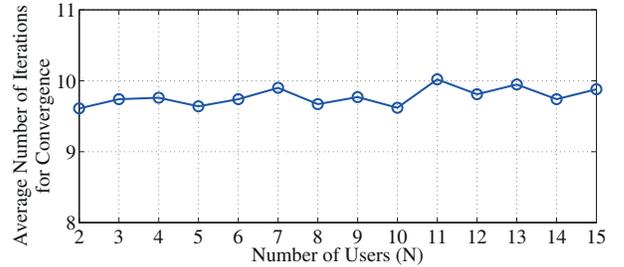} 
        \caption{Average number of iterations till convergence for $\bar{R}_{\mathrm{p}} = \mathrm{1.5}$~$\mathrm{nats/s/Hz}$ and $P_{\mathrm{e}} = 20$ dBm.}
      \label{fig:compl}
    \end{figure}
The following proposition states the properties of the optimal solution.
\begin{proposition}
The optimal solution $(\boldsymbol{E}^{*}_{\boldsymbol{\mathrm{s}\mathrm{h}}}, \boldsymbol{E}^{*}_{\boldsymbol{\mathrm{s}\mathrm{p}}}, \boldsymbol{t}^{*})$ satisfies the constraints \eqref{eq:qos}, \eqref{eq:conv1}, and \eqref{eq:conv2} with equality.
\label{prop1}
\end{proposition}
\begin{proof}
Assume that the constraints \eqref{eq:qos} and \eqref{eq:conv1} are satisfied with strict inequality. For \eqref{eq:qos}, we can decrease $E_{i\mathrm{p}}$ on the relaying link without violating the primary rate constraint and increase $E_{i\mathrm{h}}$ on the secondary access link to HAP. This increases SU sum-throughput, which contradicts with our assumption of optimality. Similarly, for \eqref{eq:conv1}, we can always increase $E_{i\mathrm{h}}$ without violating the constraint, resulting in higher sum-throughput, contradicting our assumption. Now, let us assume \eqref{eq:conv2} is satisfied with inequality. Then, we can increase $t_i$ improving SU sum-throughput, contradicting with our assumption. 
\end{proof}

\subsection{SU Selection for Relaying}
We now use structural properties of the optimal solution to find the underlying SU selection to relay PU data. We proceed in this direction by first proving the following lemma, which will be used in the proof of Proposition~\ref{Esp_ratio}.
\begin{lemma}
For SU$_i$ and SU$_j$ transmitting in the rewarded period, we have $\frac{\gamma_{i\mathrm{h}}E^{*}_{i\mathrm{h}}}{t^{*}_{i}} = \frac{\gamma_{j\mathrm{h}}E^{*}_{j\mathrm{h}}}{t^{*}_{j}}$ with $ i \neq j$.
\label{monot_lemma}
\end{lemma}
\begin{proof}
Rearranging terms of~\eqref{eq:ta}, we obtain
\begin{equation}
\exp\left( \mathcal{W}\left( \frac{-1}{\nu^{*} + 1}  \right) + \nu^{*} +1\right)-1  = \frac{{\gamma_{i\mathrm{h}}}E^{*}_{i\mathrm{h}}}{t^{*}_{i}},\quad \forall i,
\end{equation}
Since the term $\exp\left( \mathcal{W}\left( \frac{-1}{\nu^{*} + 1}  \right) + \nu^{*} +1\right)-1 $ is same for all SUs who access the spectrum, we have $\frac{\gamma_{i\mathrm{h}}E^{*}_{i\mathrm{h}}}{t^{*}_{i}} = \frac{\gamma_{j\mathrm{h}}E^{*}_{j\mathrm{h}}}{t^{*}_{j}} = \frac{1}{C}$, where $C$ is a constant.
\end{proof}
For STORA scheme, the following proposition states the underlying SU selection to relay PU data and the energy allocation at the relaying SUs.
\begin{proposition}
1) The SUs are prioritized to relay primary data in the increasing order of ratio $\frac{\gamma_{i\mathrm{h}}}{ \gamma_{i\mathrm{p}}}$. 2) The optimal energy allocated on the relaying link $E^*_{i\mathrm{p}}$ is fractional for at most one SU.
\label{Esp_ratio}
\end{proposition}
\begin{proof}
\textit{1)} SUs that relay PU data have to allocate their harvested energy between relaying and transmitting their own data. The energy allocation subproblem $\mathtt{SP1}$ of the problem~\eqref{prob1} given by \eqref{eq:sp1} aims to maximize the objective~\eqref{eq:obj} under the primary rate constraint~\eqref{eq:qos} and the energy neutrality constraint~\eqref{eq:conv1}. Under the optimal condition, since the constraints~\eqref{eq:qos} and \eqref{eq:conv1} are met with equality (see Proposition~\ref{prop1}), we can write them as
\begin{align}
\sum_{\mathrm{SU}_i \in \mathcal{S_{D}}}\gamma_{i\mathrm{p}}E_{i\mathrm{p}}&= t_0\left(\exp\left(\frac{\bar{R}_{\mathrm{p}} - Q_1t_{\mathrm{e}}}{t_0}\right) - (1+\gamma_{\mathrm{p}})\right),\label{eq:mod_sum} \\
E_{i\mathrm{p}} + E_{i\mathrm{h}} &= (P_{\mathrm{e}} + \theta_i) t_{\mathrm{e}},\,\,\,\, \forall i,\label{eq:mod_conv1}
\end{align}
respectively. Without altering the solution, we can make $\gamma_{i\mathrm{p}}$ in \eqref{eq:mod_sum} equal to a constant $K$ and scale corresponding ${\gamma_{i\mathrm{h}}}$, $E_{i\mathrm{h}}$, and $E_{i\mathrm{p}}$ accordingly. Then, the energy allocation subproblem $\mathtt{SP1}$ can be reformulated as
\begin{subequations}
\begin{align}
\! \mathop{\mathrm{maximize}}_{\boldsymbol{\bar{E}_{\mathrm{s}\mathrm{h}}},\boldsymbol{\bar{E}_{\mathrm{s}\mathrm{p}}}}
& ~~ \sum_{i=1}^N t_{i}\ln\left(1+ \frac{\bar{\gamma}_{i\mathrm{h}}\bar{E}_{i\mathrm{h}} }{t_{i}} \right) \label{eq:mod_obj1} \\
\mathrm{subject~to} &~~ K\!\!\!\! \sum_{\mathrm{SU}_i \in \mathcal{S_{D}}}\!\!\!\!\bar{E}_{i\mathrm{p}} = t_0\left(\!\exp\left(\frac{\bar{R}_{\mathrm{p}} - Q_1t_{\mathrm{e}}}{t_0}\!\right) - (1+\gamma_{\mathrm{p}})\right)\!,\label{eq:mod_sum1} \\
&~~ \frac{K}{\gamma_{i\mathrm{p}}}(\bar{E}_{i\mathrm{p}}  + \bar{E}_{i\mathrm{h}}) = (P_{\mathrm{e}}+\theta_i)t_{\mathrm{e}} , \hspace{3mm} \forall i,\label{eq:con11} \\
&~~ \bar{E}_{i\mathrm{p}}, \bar{E}_{i\mathrm{h}}  \geq 0, ~~~\forall  i, \label{eq:con12} 
\end{align}
\label{eq:fin_eq}
\end{subequations}\vspace*{-4mm}

\noindent where $ \bar{\gamma}_{i\mathrm{h}} = \gamma_{i\mathrm{h}}\frac{K}{\gamma_{i\mathrm{p}}},  \bar{E}_{i\mathrm{p}} = E_{i\mathrm{p}}\frac{\gamma_{i\mathrm{p}}}{K}$, $ \bar{E}_{i\mathrm{h}} = E_{i\mathrm{h}}\frac{\gamma_{i\mathrm{p}}}{K}$, $\boldsymbol{\bar{E}_{\mathrm{s}\mathrm{h}}} = [\bar{E}_{1\mathrm{h}}, \dotsc, \bar{E}_{N\mathrm{h}}]$, and $\boldsymbol{\bar{E}_{\mathrm{s}\mathrm{p}}} = [\bar{E}_{1\mathrm{p}}, \dotsc, \bar{E}_{N\mathrm{p}}]$. Using \eqref{eq:mod_obj1}, the SU sum-throughput can be rewritten as
\begin{align}
R_{\mathrm{s,sum}} =  \sum_{i=1}^N t_{i}\ln\left(1+ \frac{K \frac{\gamma_{i\mathrm{h}}}{\gamma_{i\mathrm{p}}}\bar{E}_{i\mathrm{h}}}{t_{i}} \right).\label{eq:inter}
\end{align}
Using Lemma~\ref{monot_lemma}, we have $\frac{\gamma_{i\mathrm{h}}E_{i\mathrm{h}}}{t_{i}} = \frac{\bar{\gamma}_{i\mathrm{h}}\bar{E}_{i\mathrm{h}}}{t_{i}} = \frac{1}{C}$, where $C$ is a constant. Thus, we can write \eqref{eq:inter} as
\begin{align}
R_{\mathrm{s,sum}} = \sum_{i=1}^N CK \frac{ \gamma_{i\mathrm{h}}}{\gamma_{i\mathrm{p}}} \bar{E}_{i\mathrm{h}} \ln\left(1+ \frac{1}{C} \right).
\label{eq:Rs_sum_ratio}
\end{align}
To exploit the structural properties of the optimization problem \eqref{eq:fin_eq} with its objective \eqref{eq:mod_obj1} replaced by \eqref{eq:Rs_sum_ratio}, we write it in the form of a generic resource allocation problem as follows:
\begin{subequations}
\begin{align}
  \mathop{\mathrm{maximize}}_{\boldsymbol{p}, \boldsymbol{q}}&~~~ \sum_{i = 1}^{N} w_{i}p_{i} \label{eq:obj_generic}\\
\mathrm{subject~to} &~~~ p_{i} + q_{i}  =  r_{i}, \label{eq:sum} \\
&~~~ \sum_{i = 1}^{N} q_{i} = s, \label{eq:sum1}  \\ 
&~~~\boldsymbol{p}, \boldsymbol{q} \geq 0,\label{eq:final}
\end{align}
\label{eq:prob_lemma}\vspace*{-5mm}
\end{subequations}

\noindent where $w_{i} = CK \frac{ \gamma_{i\mathrm{h}}}{\gamma_{i\mathrm{p}}} \ln\left(1+ \frac{1}{C} \right)$, $p_i = \bar{E}_{i\mathrm{h}}$, $q_i = \bar{E}_{i\mathrm{p}}$, $r_{i} = \left(P_{\mathrm{e}}+ \theta_i \right)t_{\mathrm{e}}\frac{\gamma_{i\mathrm{p}}}{K}$, $s = \frac{t_0}{K}\left(\exp\left(\frac{\bar{R}_{\mathrm{p}} - Q_1t_{\mathrm{e}}}{t_0}\right) - (1+\gamma_{\mathrm{p}})\right)$, $\boldsymbol{p} = [p_{1}, \dotsc, p_{N}]$, and $\boldsymbol{q} = [q_{1}, \dotsc, q_{N}]$.

We split the resource allocation problem in~\eqref{eq:prob_lemma} into three cases to obtain its optimal solution. 
Without loss of generality, let us assume that $w_{1} > w_{2} > \dotsc > w_{N}$. We denote $r' = \min(r_{1}, r_2, \dotsc, r_N)$. 

\underline{\textit{Case 1}}: Let $ s < r' $.
Consider following two possible strategies:
\begin{itemize}
\item[1.] Strategy 1: $q_{i} = s$, $q_{j} = 0$ with $j \neq i$.
\item[2.] Strategy 2: $q_{i} = a_{i}s$. Here, $a_i$ ($0 \leq a_i < 1$) can be chosen arbitrarily such that $\sum_{i=1}^{N} a_{i} = 1$.
\end{itemize}
Under strategy 1, if SU$_1$ allocates $s$ on $\mathbb{Q}$ link, the objective of problem~\eqref{eq:prob_lemma} can be rewritten as
\begin{align}
\sum_{i = 1}^{N} w_{i}p_{i} &=  w_{1}(r_{1}-s) +  {\underset{ j = 2} {\overset{N} \sum}} w_{j} r_j \nonumber \\
&=  \sum_{i=1}^{N}w_i r_{i} - w_1 s \nonumber  \\
&<  \sum_{i=1}^{N}w_i r_{i} - w_2 s \nonumber  \\
&\hspace*{7mm} \vdots \hspace{9mm}  \vdots  \hspace{9mm} \vdots \nonumber  \\ 
& < \sum_{i=1}^{N}w_i r_{i} - w_N s.
\label{eq:strat1}
\end{align}
Thus, allocating the resource $s$ fully on $\mathbb{Q}$ link of SU$_N$ with the lowest weight $w_{N}$, i.e., when $q_{N} = s$, maximizes the objective. Now, under strategy 2, we can write the objective as
\begin{align}
\sum_{i = 1}^{N} w_{i}p_{i} & =  \sum_{i=1}^{N} w_{i}(r_{i} -a_{i}s) \nonumber \\
& = \sum_{i = 1}^{N}w_i r_i - \sum_{i = 1}^{N}a_i w_i s.
\label{eq:strat2}
\end{align}
Since $\sum_{i = 1}^{N}a_i w_i s > \sum_{i = 1}^{N}a_iw_N s = w_N s$, we have $\eqref{eq:strat1} >  \eqref{eq:strat2}$. Thus, strategy 1 outperforms strategy 2 and allocating complete resource $s$ to SU with the lowest weight yields the optimal solution, i.e., $q^*_N = s$ and $q^*_{i}$ = 0 if $i \neq N$. From \eqref{eq:sum}, we then have $p^*_N = r_N - s$ and $p^*_{i} = r_i$ if $i \neq N$.\footnote{If $w_i = w_j = \min{\boldsymbol{w}}$ with $i \neq j$, then we can randomly allocate resource $s$ completely to either SU$_i$ or SU$_j$.}

\underline{\textit{Case 2}}: Let $r' < s \leq r_N $. Under strategy 1, for SU$_i$ with $r_{i} < s$, we cannot allocate complete resource $s$ as it makes $p_i < 0$. Similar argument applies for strategy 2 and when $r_{i} < a_i s $. Then, excluding such users and following the same proof for \textit{Case 1}, we can show that the objective is maximized if user $N$ with the lowest weight is allocated the resource $s$ fully. Therefore, $p^*_i$ and $q^{*}_{i}$ are the same as those for \textit{Case 1}.

\underline{\textit{Case 3}}: Let $s > r_N$. In this case, the $\mathbb{Q}$ link of SU$_N$ can accommodate at the most $r_N$ out of the resource $s$. This implies that we have to allocate the resource $s$ over $\mathbb{Q}$ links of more than one user. Following the proof given for \textit{Case 1}, we can show that $q^{*}_{N} = r_N$ and $ p^{*}_{N} = 0$.  Now, the problem~\eqref{eq:prob_lemma} reduces to finding $\boldsymbol{p}$ and $\boldsymbol{q}$ of dimension $N-1$, i.e., with $N-1$ users. In this modified optimization problem, $s$ in \eqref{eq:sum1} is replaced by $s-r_N$. Here, we have to investigate two cases depending on whether $(s-r_N) \leq r_{N-1}$ or $(s-r_N) > r_{N-1}$, to obtain the optimal solution of the modified optimization problem. Based on the aforementioned discussion in this proof, we get $q^{*}_{N-1} = \min(r_{N-1}, s-r_N)$ and $ p^{*}_{N-1} = r_{N-1}-q^*_{N-1}$. The problem~\eqref{eq:prob_lemma} can be further reduced to dimension $(N-2)$ and later upto $(N-(N-1))$ recursively. Combining the solution of each recursive step together, we can finally write the optimal solution to the problem \eqref{eq:prob_lemma} as 
\begin{eqnarray}
    q_{i}^* =
\begin{cases}
     \min(r_{i},s),& \!\! \hspace*{-9mm}  \text{if}~ w_{i} = \min (\boldsymbol{w}), \\
       \min\left(\!\!r_{i}, \left[\!s - {\underset{\lbrace j:   j \neq i, w_{j} \leq w_{i}\rbrace} \sum} \!q_{j}^{*}\!\right]^{+}\!\right)\!, &  \!\!\!\text{otherwise},
\end{cases} 
\label{eq:sol_p}
\end{eqnarray}
and
\begin{align}
p_{i}^* = r_{i} - q^{*}_{i}, \hspace{2mm} \forall i,
\label{eq:sol_q}
\end{align}
which also comprises the optimal solutions for cases 1 and 2.

{The optimal solution of \eqref{eq:prob_lemma} given in \eqref{eq:sol_p} and \eqref{eq:sol_q} along with its derivation shows that the objective~\eqref{eq:obj_generic} is maximized when the resources $q_i$ on $\mathbb{Q}$ links of SUs are allocated in the increasing order of the weights $w_i$ on their corresponding $\mathbb{P}$ links.} Since the subproblem \eqref{eq:fin_eq} is identical to \eqref{eq:prob_lemma}, the SU sum-throughput given by \eqref{eq:Rs_sum_ratio} under the constraints \eqref{eq:mod_sum1}-\eqref{eq:con12} is maximized when the SU with the lowest $w_i = CK \frac{ \gamma_{i\mathrm{h}}}{\gamma_{i\mathrm{p}}} \ln\left(1+ \frac{1}{C} \right)$, i.e., $\frac{ \gamma_{i\mathrm{h}}}{\gamma_{i\mathrm{p}}}$ (since $CK \ln\left(1+ \frac{1}{C} \right)$ is same for all SUs), is prioritized to relay PU data and then the SU with the next lowest $\frac{ \gamma_{i\mathrm{h}}}{\gamma_{i\mathrm{p}}}$. This completes the proof of the first part of Proposition~\ref{Esp_ratio}.

\textit{2)} If the cooperation from SU with the lowest $\frac{ \gamma_{i\mathrm{h}}}{\gamma_{i\mathrm{p}}}$ to relay PU data cannot meet the primary rate constraint even after it has spent its complete harvested energy on the relaying link, then only the SU with the next lowest $\frac{ \gamma_{i\mathrm{h}}}{\gamma_{i\mathrm{p}}}$ is considered to relay PU data, and so on. Thus, at most one SU will spend a fraction of its harvested energy on the relaying link, while other relaying SUs spend their complete harvested energy to relay PU data.
\end{proof}

The following remark highlights the unfairness created among SUs by the STORA scheme.
\begin{remark}
Relaying SUs that spend their complete harvested energy in relaying achieve zero throughput. On the contrary, SUs that do not relay PU data gain spectrum access and use their complete harvested energy to achieve non-zero individual throughput. Thus, the STORA scheme creates unfairness among SUs in terms of individual throughput. 
\end{remark}
\section{Fairness Enhancing Resource Allocation}
\label{sec:fair_schemes}
In this section, we examine three fairness enhancing schemes: (i) \textit{equal time allocation} (ETA), (ii) \textit{minimum throughput maximization} (MTM), and (iii) \textit{proportional time allocation} (PTA).

\subsection{Equal Time Allocation}
The goal is to maximize the SU sum-throughput under the condition of allocating all SUs equal time to access the channel, i.e., $t_{1} = t_{2}  = \dotsc = t_{N} = t_{\mathrm{eq}}$. Unlike STORA scheme, the condition of equal time to each SU  ensures that all SUs get an opportunity to access the channel.
For a fixed decoding set $\mathcal{S_D}$, the problem of SU sum-throughput maximization with equal time allocation is formulated as follows:
\begin{align}
\mathop{\mathrm{maximize}}_{\boldsymbol{E_{\mathrm{s}\mathrm{h}}},\boldsymbol{E_{\mathrm{s}\mathrm{p}}}, t_{\mathrm{e}}, t_{0}, t_{\mathrm{eq}}}& ~~~ \sum_{i=1}^N t_{\mathrm{eq}}\ln\left(1+ \frac{\gamma_{i\mathrm{h}}E_{i\mathrm{h}} }{t_{\mathrm{eq}}} \right)  \nonumber \\
\mathrm{subject~to} &~~~ \eqref{eq:qos} - \eqref{eq:conv1}  \nonumber \\
&~~~ t_{\mathrm{e}}+2t_{0}+Nt_{\mathrm{eq}} \leq 1,  \nonumber \\
 & ~~~\boldsymbol{E_{\mathrm{s}\mathrm{h}}},\boldsymbol{E_{\mathrm{s}\mathrm{p}}}, t_{\mathrm{e}}, t_{0}, t_{\mathrm{eq}} \geq 0.  
\label{eq:eta}
\end{align}
The problem in~\eqref{eq:eta} is a convex optimization problem and Proposition~\ref{prop1} also holds true for \eqref{eq:eta}. The STORA and ETA problems differ only in their allocation of the access time. Thus, we can use the same approach as that of STORA scheme to solve~\eqref{eq:eta} with the optimal $t_{\mathrm{eq}}$ given in the following proposition.
\begin{proposition}
\label{eta_teq}
The optimal spectrum access time $t^{*}_{\mathrm{eq}}$ of an SU is the solution of the equation
\begin{equation}
\sum_{i=1}^{N} \ln\left(1+ \frac{\gamma_{i\mathrm{h}}E^{*}_{i\mathrm{h}} }{t^{*}_{\mathrm{eq}}} \right) - \frac{\frac{\gamma_{i\mathrm{h}}E^{*}_{i\mathrm{h}} }{t^{*}_{\mathrm{eq}}}  }{1+\frac{\gamma_{i\mathrm{h}}E^{*}_{i\mathrm{h}} }{t^{*}_{\mathrm{eq}}}} = N\nu^{*}.
\label{eq:teq}
\end{equation}
\end{proposition}
\begin{proof}
See Appendix~\ref{app:prop_eta}.
\end{proof}
As~\eqref{eq:teq} shows, we cannot obtain the closed-form expression for $t^{*}_{\mathrm{eq}}$. But, we can find $t^{*}_{\mathrm{eq}}$ numerically using bisection method. Interestingly, an asymptotic solution in closed-form exists when HAP experiences high received SNR conditions, i.e., when $\frac{\gamma_{i\mathrm{h}}E^{*}_{i\mathrm{h}} }{t^{*}_{\mathrm{eq}}} \gg 1$. In this case, \eqref{eq:teq} can be written as
 \begin{equation}
\sum_{i=1}^{N} \ln\left(\frac{\gamma_{i\mathrm{h}}E^{*}_{i\mathrm{h}} }{t^{*}_{\mathrm{eq}}} \right) = N\nu^{*},
\label{eq:teq_approx}
\end{equation}
which yields $ \sum_{i=1}^{N} \ln\left({\gamma_{i\mathrm{h}}E^{*}_{i\mathrm{h}} } \right) = N\nu^{*} + N\ln(t^*_{\mathrm{eq}})$. Rearranging terms, we get
\begin{align}
 t^*_{\mathrm{eq}} & = \frac{1}{\exp(\nu^{*})}\left(\prod_{i= 1}^{N} \gamma_{i\mathrm{h}}E^{*}_{i\mathrm{h}}\right)^{\frac{1}{N}}.
 \label{eq:teq_high_snr1}
 \end{align}
Thus, \eqref{eq:teq_high_snr1} shows that, under high SNR conditions, the time allocated for spectrum access to an SU is proportional to the geometric mean of the received energies at HAP. Substituting \eqref{eq:teq_high_snr1} in \eqref{eq:teq_approx}, the objective of~\eqref{eq:eta}, i.e., SU sum-throughout, can be written as
\begin{equation}
C_{2} \left(\prod_{i= 1}^{N} \gamma_{i\mathrm{h}}E^{*}_{i\mathrm{h}}\right)^{\frac{1}{N}} \ln\left(\frac{1}{C_{2}}\left(\prod_{i= 1}^{N} \gamma_{i\mathrm{h}}E^{*}_{i\mathrm{h}}\right)^{\frac{N-1}{N}}\right),
\label{eq:mod_obj_eq}
\end{equation}
where $C_{2} = \frac{1}{\exp(\nu^{*})}$. Thus, even if a single SU is allocated zero energy on the link to HAP, the objective \eqref{eq:mod_obj_eq} becomes zero. Therefore, each SU attains zero throughput. This implies that no SU should spend full amount of harvested energy relaying PU data in order to achieve non-zero SU sum-throughput.

Let us now consider the case where SUs experience bad channel conditions to HAP, i.e., a low SNR scenario.  Using the approximation $\ln(1+x) \approx x$ for $x \ll 1 $, the objective in~\eqref{eq:eta} becomes $\sum_{i=1}^{N} \gamma_{i\mathrm{h}}E_{i\mathrm{h}}$. Then, using optimality conditions given in Proposition \ref{prop1}, the energy allocation subproblem is written as
\begin{align}
 \mathop{\mathrm{maximize}}_{\boldsymbol{E_{\mathrm{s}\mathrm{h}}},\boldsymbol{E_{\mathrm{s}\mathrm{p}}}} &~~~ \sum_{i=1}^{N} \gamma_{i\mathrm{h}}E_{i\mathrm{h}} \nonumber \\
\mathrm{subject~to}&~~~ \eqref{eq:mod_sum}, \eqref{eq:mod_conv1}, \boldsymbol{E_{\mathrm{s}\mathrm{h}}},\boldsymbol{E_{\mathrm{s}\mathrm{p}}} \geq 0,
\end{align}
which is equivalent to the problem in \eqref{eq:prob_lemma}.
Then, the energy allocation in ETA scheme in low SNR regime is same as the energy allocation in STORA scheme, i.e., SUs are prioritized to relay PU data in the increasing order of the ratio $\frac{\gamma_{i\mathrm{h}}}{\gamma_{i\mathrm{p}}}$, and at most one SU spends fractional energy on the relaying link. This implies that an SU might not have residual energy for spectrum access, despite being allocated time for it. Thus, in low SNR regime, even though the relay and energy allocation is same as that of STORA scheme, the unused time due to the lack of energy causes lower sum-throughput in ETA scheme compared to STORA scheme. {Since the problem in \eqref{eq:eta} has the same structure as the problem \eqref{prob1} for STORA scheme, we can use Algorithms~\ref{main_alg} and \ref{alg1} to compute the optimal solution of  \eqref{eq:eta}, which has the same complexity as STORA scheme, i.e., $\mathcal{O}(N^2)$.}

\subsection{Minimum Throughput Maximization (MTM)}
\label{sec:mmf}
The aim is to guarantee minimum throughput to each SU and yet maximize the SU sum-throughput. For a fixed $\mathcal{S_{D}}$, the throughput maximization problem for MTM scheme is given by
\begin{subequations}
\begin{align}
\mathop{\mathrm{maximize}}_{\boldsymbol{E_{\mathrm{s}\mathrm{h}}},\boldsymbol{E_{\mathrm{s}\mathrm{p}}}, \boldsymbol{t}, R_{\min}}&~~~R_{\min}  \nonumber \\
\mathrm{subject~to} & ~~~ t_{{i}}\ln\left(1+ \frac{\gamma_{i\mathrm{h}}E_{i\mathrm{h}} }{t_{i}} \right) \geq R_{\min}, \label{eq:mmf_con} \\
&~~~ \eqref{eq:qos} - \eqref{eq:last},  R_{\min} \geq 0.
\end{align}
\label{eq:mmf}
\end{subequations}
Since the term $t_{{i}}\ln\left(1+ \frac{\gamma_{i\mathrm{h}}E_{i\mathrm{h}} }{t_{i}} \right)$ in \eqref{eq:mmf_con} monotonically increases with $t_i$ and $E_{i\mathrm{h}}$, \eqref{eq:mmf_con} implies that the maximum $R_{\min}$ is obtained when rates of all SUs are equal. Thus, unless $ R_{\min} = 0$, all SUs get a chance to transmit on the direct link to HAP and no SU spends full amount of harvested energy on the relaying link. The SUs that have not decoded PU data successfully can access the spectrum, and thus Proposition~\ref{prop1} holds true. 

The problem~\eqref{eq:mmf} is convex and Slater's condition~\cite[Chapter 3]{bertesekas} holds true. Thus, the analytical expressions of the optimization variables can be obtained using KKT conditions.
Similar to the STORA problem, for a fixed $\mathcal{S_D}$, this problem can be divided into two levels, one for solving the energy allocation and time allocation on secondary links and other for solving $t_{\mathrm{e}}$ and $R_{\min}$ as a master primal problem. 
The Lagrangian dual of the energy allocation subproblem for a fixed $(R_{\min}, t_{\mathrm{e}})$ is given by
\begin{align}
\underset{\boldsymbol{E_{\mathrm{s}\mathrm{h}}},\boldsymbol{E_{\mathrm{s}\mathrm{p}}}}{\mathop{\mathrm{max}}} \hspace{2mm} & R_{\min} - \sum_{i=1}^N \rho_{i}\left(R_{\min} - R_{i}(E_{i\mathrm{h}}, t_{i})\right) \nonumber \\
&- \lambda \left( \bar{R}_{\mathrm{p}} -  R_{\mathrm{p,c}}(\boldsymbol{E_{\mathrm{sp}}} , t_{0})\right) \nonumber \\
&- \sum_{i = 1}^{N} \mu_{i}(E_{i\mathrm{p}} + E_{i\mathrm{h}} - (P_{\mathrm{e}} + \theta_i) t_{\mathrm{e}}), 
\label{eq:dual1_mmf}
\end{align}

\noindent where $\rho_{i}$ is the dual variable associated with the secondary rate constraint~\eqref{eq:mmf_con} and
dual of the time allocation subproblem is
\begin{align}
\underset{t_{0}, \boldsymbol{t_{\mathrm{a}}}}{\mathop{\mathrm{max}}} \hspace{2mm} &  R_{\min} - \sum_{i=1}^N \rho_{i}\left(R_{\min} - R_{i}(E_{i\mathrm{h}}, t_{i})\right) - \lambda \bar{R}_{\mathrm{p}}  \nonumber \\
& + \lambda R_{\mathrm{p,c}}(\boldsymbol{E_{\mathrm{sp}}} , t_{\mathrm{e}}, t_{0})  
-  \kappa \left( \bar{R}_{\mathrm{p}} - Q_1 t_{\mathrm{e}} - Q_2 t_{0} \right) \nonumber \\
& -  \nu (t_{\mathrm{e}}+2t_{0}+\sum_{i=1}^{N}t_{i}-1).
\label{eq:dual2_mmf}
\end{align}
The dual variables $(\lambda, \kappa, \nu)$ can be computed using their gradients given by \eqref{eq:grad_lam}-\eqref{eq:grad_nu}, while the dual variable $\boldsymbol{\mu}$ is obtained using bisection method. The dual variable $\boldsymbol{\rho} = [\rho_1, \rho_2, \dotsc, \rho_N]$ minimizing the dual function \eqref{eq:dual2_mmf} can be found using its gradient given by
$ g_{\rho_{i}} = R_{i}(E_{i\mathrm{h}}, t_{i}) - R_{\min}$.
The master primal problem in charge of updating $t_{\mathrm{e}}$ and $R_{\min}$ can be solved using subgradient based methods like ellipsoid method~\cite[Chapter 1]{bertesekas} with their subgradients given by \eqref{eq:grad_te} and $g_{R_{\min}} = 1 - \sum_{i=1}^N \rho_i$, respectively. The Levels 1 and 2 are executed in a similar fashion as that of Algorithm $\ref{alg1}$ with additional dual update using $g_{{\boldsymbol{\rho}}}$ in Level 1 and primal update using $g_{R_{\min}}$ in Level 2. The decoding set that maximizes the sum-throughput can be found using the strategy presented in Section \ref{sec:opt_dec_set} and Algorithm~\ref{main_alg}.

The following proposition provides the analytical expressions corresponding to the optimal solution.
\begin{proposition}
The optimal energy and time allocation for the MTM scheme is given by
\begin{align*}
& [E^{*}_{\mathrm{s}_{i}\mathrm{h}}, ~E^{*}_{i\mathrm{p}}] =  \left[t^{*}_{i}\left[ \frac{\rho^{*}_{i}}{\mu^{*}_{i}} - \frac{1}{\gamma_{i\mathrm{h}}}\right]^{+}, \eqref{eq:Ep} \right], \nonumber \\
& [t^{*}_{\mathrm{e}}, ~t^{*}_{0}, ~t^{*}_{i}] = \!\left[\eqref{eq:te}, ~\eqref{eq:t0}, ~\frac{{\gamma_{i\mathrm{h}}}E^{*}_{i\mathrm{h}}}{\exp \left({ \mathcal{W}\left( \frac{-1}{\frac{\nu^{*}}{\rho^*_{i}} + 1}  \right) + \frac{\nu^{*}}{\rho^*_{i}} +1 }\right)-1} \right]\!\!.
\end{align*}
\label{opt_soln_mmf}
\end{proposition}
\begin{proof}
See Appendix \ref{app_prop_mmf}.
\end{proof}
{\textit{Complexity of MTM scheme}: 
Algorithm~\ref{alg1} for MTM scheme needs to update $N+3$ dual variables. Updating these variables using the interior point method has complexity $\mathcal{O}(N^3)$~\cite[Chapter 4]{bertesekas}, which dominates the complexity of Algorithm~\ref{alg1} for MTM scheme. Including the complexity of Algorithm~\ref{main_alg}, the complexity to calculate the optimal solution for MTM scheme is $\mathcal{O}(N^4)$.}
\subsection{Proportional Time Allocation (PTA)}
Recall that the cooperation between PU and SUs improves the received primary SNR and reduces the PU transmission time for a target primary rate. This, in turn, creates an opportunity for SUs to transmit their own data to HAP. Thus, to divide the reward-generated access time among SUs, one of the plausible criteria is to allocate the access time to an SU proportional to its contribution in the received primary SNR. Then, in this case, to obtain an opportunity to access the spectrum, an SU has to relay PU data. Now, to relay PU data, the SU should decode PU data successfully. Based on this \textit{decode-relay-then-access} condition, the optimization problem for the PTA scheme is formulated for a fixed $\mathcal{S_D}$ as
\begin{subequations}
\begin{align}
\mathop{\mathrm{maximize}}_{\boldsymbol{E_{\mathrm{s}\mathrm{h}}},\boldsymbol{E_{\mathrm{s}\mathrm{p}}}, \boldsymbol{t}} &~~~ \sum_{i=1}^N t_{i}\ln\left(1+ \frac{\gamma_{i\mathrm{h}}E_{i\mathrm{h}} }{t_{i}} \right)  \label{eq:obj22} \\
\mathrm{subject~to}&~~~ t_{i} = \sum_{j = 1}^{N} t_{j} \frac{\gamma_{i\mathrm{p}}E_{i\mathrm{p}}}{\sum_{j =1}^{N}\gamma_{j\mathrm{p}}E_{j\mathrm{p}}}, \label{eq:prop_time} \\ 
&~~~\eqref{eq:qos} - \eqref{eq:last}, \label{eq:common_const}
\end{align}
\label{eq:prob4}
\end{subequations}\vspace*{-4mm}

\noindent where \eqref{eq:prop_time} is the proportional time allocation constraint. Let us define the ratio
\begin{equation}
 \frac{\sum_{j = 1}^{N} t_{j}}{\sum_{j =1}^{N}\gamma_{j\mathrm{p}}E_{j\mathrm{p}}} = \frac{t_{i}}{\gamma_{i\mathrm{p}}E_{i\mathrm{p}}}  = \zeta,
 \label{eq:prop_ratio}
\end{equation}
where $\zeta$ is a positive constant whose optimal value is unknown. Equation \eqref{eq:prop_ratio} implies $t_{i} = \zeta \gamma_{i\mathrm{p}}E_{i\mathrm{p}}$ and $\sum_{j = 1}^{N} t_{j} = \zeta \sum_{j =1}^{N}\gamma_{j\mathrm{p}}E_{j\mathrm{p}}$; nevertheless, the former includes the latter. Therefore, the constraint~\eqref{eq:prop_time} can be replaced by
\begin{equation}
t_{i} = \zeta \gamma_{i\mathrm{p}}E_{i\mathrm{p}}. 
\label{eq:prop_fair_mod} 
\end{equation}
Observe that the product of optimization variables $\zeta$ and $E_{i\mathrm{p}}$, for all $i$, in \eqref{eq:prop_fair_mod} makes the problem \eqref{eq:prob4} non-convex. However, for a given $\zeta$, the constraint \eqref{eq:prop_fair_mod} is affine. The other constraints which are independent of $\zeta$ are either concave or affine as shown for STORA, ETA, and MTM schemes. Therefore, for a given $\zeta$, the problem \eqref{eq:prob4} is a convex optimization problem. The optimal decoding set $\mathcal{S_{D}}$ can be obtained using the strategy proposed in section \ref{sec:opt_dec_set} and Algorithm~\ref{main_alg}. Using \eqref{eq:prop_fair_mod}, for a fixed $(\mathcal{S_D}, \zeta, t_{\mathrm{e}})$, the energy allocation and time allocation subproblems are given by
\begin{align}
\mathtt{SP3}:~ \mathop{\mathrm{maximize}}_{\boldsymbol{{E}_{{\mathrm{s}\mathrm{h}}}}, \boldsymbol{{E}_{{\mathrm{s}\mathrm{p}}}}} &~~~\sum_{i=1}^N R_{i}(E_{i\mathrm{h}}, \zeta \gamma_{i\mathrm{p}}E_{i\mathrm{p}} ) \nonumber \\
\mathrm{subject~to} &~~~ \eqref{eq:qos}, \eqref{eq:conv1}, \eqref{eq:last}
\label{eq:prop_sp1}
\end{align}
and
\begin{align}
\mathtt{SP4}:~ \mathop{\mathrm{maximize}}_{t_{0}}&~~~ \sum_{i=1}^N R_{i}(E_{i\mathrm{h}}, \zeta \gamma_{i\mathrm{p}}E_{i\mathrm{p}} ) \nonumber \\
\mathrm{subject~to}&~~~ \eqref{eq:qos}, \eqref{eq:dec_con2}, \nonumber \\
& ~~~  t_{\mathrm{e}} + 2t_{0} + \sum_{i=1}^{N}\zeta \gamma_{i\mathrm{p}}E_{i\mathrm{p}} \leq 1, \nonumber \\
&~~~ t_{0} \geq 0,
\label{eq:prop_sp2} 
\end{align}
respectively. The time allocation subproblem $\mathtt{SP4}$ reduces to a feasibility problem to solve for $t_{0}$ than solving both $(\boldsymbol{t_{\mathrm{a}}}, t_{0})$, and the master problem is responsible for updating $t_{\mathrm{e}}$ as that of the STORA scheme. 
For a given $\zeta$, the subproblems $\mathtt{SP3}$ and $\mathtt{SP4}$, and the master problem are solved iteratively in a similar manner as that of in Algorithm \ref{alg1}. Since strong duality holds for a fixed $\zeta$, the optimization variables can be found using KKT conditions as presented in Proposition \ref{opt_soln_prop}.
\noindent To gain insights on the effect of $\zeta$ on the SU sum-throughput, we obtain the Lagrangian of the problem \eqref{eq:prob4} 
\begin{align}
\mathcal{L}_{4}  &=  \sum_{i=1}^N R_{i}(E_{i\mathrm{h}}, \zeta\gamma_{i\mathrm{p}}E_{i\mathrm{p}} ) 
 - \lambda \left(\bar{R}_{\mathrm{p}} - R_{\mathrm{p,c}}\left(\boldsymbol{E_{\mathrm{s}\mathrm{p}}} ,t_{\mathrm{e}},  t_{0}\right) \right) \nonumber \\
 & - \sum_{i = 1}^{N} \mu_{i}(E_{i\mathrm{p}} + E_{i\mathrm{h}} - (P_{\mathrm{e}} + \theta_i)t_{\mathrm{e}}) \nonumber \\
&   - \kappa \left( \bar{R}_{\mathrm{p}} - Q_1 t_{\mathrm{e}} - Q_2 t_{0} \right) -\nu (t_{\mathrm{e}} + 2t_{0}+\sum_{i=1}^{N}\zeta\gamma_{i\mathrm{p}}E_{i\mathrm{p}}-1).
\end{align}
The Hessian of $\mathcal{L}_{4}$ with respect to $\zeta$ is
\begin{align}
\frac{\partial^2 \mathcal{L}_{4}}{\partial \zeta^2} = -\sum_{i=1}^{N}\frac{\gamma^2_{i\mathrm{h}}E^2_{i\mathrm{h}}}{\zeta^3\left( 1 + \frac{\gamma_{i\mathrm{h}}E_{i\mathrm{h}}}{\gamma_{i\mathrm{p}}E_{i\mathrm{p}}}\right)^{2} \gamma_{i\mathrm{p}}E_{i\mathrm{p}}} \leq 0 .
\end{align}
This implies that $\mathcal{L}_{4}$ is a concave function of $\zeta$. Then, the optimal value of $\zeta$ can be found using one dimensional search methods like golden section search~\cite[Appendix C]{bertesekas} over $\zeta \geq 0$. 

We give the analytical expressions corresponding to the optimal solution in the following proposition.
\begin{proposition}
The optimal solution of the proportional time allocation (PTA) scheme is
\begin{align}
\left[E^{*}_{i\mathrm{h}}, ~E^{*}_{i\mathrm{p}}\right] & = \left[\zeta^* \gamma_{i\mathrm{p}}E^{*}_{i\mathrm{p}}\left[\frac{1}{\mu_{i}} - \frac{1}{\gamma^{*}_{i\mathrm{h}}}\right]^+,  ~\mathcal{E}^{*}_{i}\right],\nonumber \\
\! \left[t^{*}_{\mathrm{e}}, ~t^{*}_{0}, ~t^*_{i} \right] & = \left[1 - 2t^{*}_{0} - \sum_{i=1}^{N} \zeta^{*}\gamma_{i\mathrm{p}}E^{*}_{i\mathrm{p}}, ~\eqref{eq:t0}, ~\zeta^* \gamma_{i\mathrm{p}}E^*_{i\mathrm{p}}\right], \nonumber 
\end{align}
where $\mathcal{E}^{*}_{i}$ is the solution of
\begin{align}
\ln\left(1+{z^{*}_{i}}\right) - \frac{z^{*}_{i}}{1+z^{*}_{i}}& = \frac{\lambda^{*}/\zeta^{*}}{1+\gamma_{\mathrm{p}}+\frac{\sum_{i=1}^{N}\gamma_{i\mathrm{p}}E^{*}_{i\mathrm{p}}}{t^{*}_{0}}} + \frac{\mu^{*}_{i}}{\zeta^*\gamma_{i\mathrm{p}}}+ \frac{\nu^{*}}{\gamma_{i\mathrm{p}}}, \nonumber
\end{align}
\label{opt_soln_prop}
with $z^{*}_{i} = \frac{\gamma_{i\mathrm{h}}E^{*}_{i\mathrm{h}}}{\zeta^{*}\gamma_{i\mathrm{p}}E^{*}_{i\mathrm{p}}}$ for all $\mathrm{SU}_i \in \mathcal{S_D}$.
\end{proposition}
\begin{proof}
See Appendix \ref{app_prop_pf}.
\end{proof}
{\textit{Complexity of PTA scheme}: 
In PTA scheme, we have additional golden section search for the variable $\zeta$, which has complexity $\mathcal{O}(\frac{1}{\epsilon})$ for the $\epsilon$-accurate solution. Since $\epsilon$ is constant ($10^{-5}$ in our case), the complexity of PTA scheme is same as that of STORA scheme, i.e, $\mathcal{O}(N^2)$.}

\section{Simulation Results and Discussions}
\label{sec:results}
We present simulation results to evaluate the performance of the proposed WP-CCRN scenario for the four resource allocation schemes, namely, STORA, ETA, MTM, and PTA schemes, against different system parameters. We also discuss the fairness achieved by each of the schemes. 

The instantaneous channel power gain of the link between user $i$ and $j$ is denoted as $h_{ij}d_{ij}^{-\beta}$, where $h_{ij}$ is the instantaneous Rayleigh channel power gain with unit mean, $d_{ij}$ is the distance between users $i$ and $j$, and $\beta $ is the path-loss exponent which is assumed to be $\mathrm{3}$. We use $P_{\mathrm{p}} = \mathrm{20~dBm}$, $\eta = \mathrm{0.5}$, $N_0 = -\mathrm{70}~\mathrm{dBm/Hz}$, and {$\Gamma = 8.8$ dB (for uncoded quadrature amplitude modulation)~\cite{goldsmith_book}}. The PT and PR are situated at $\mathrm{50}\mathrm{m}$ from each other. The locations of PT, PR, and HAP are collinear, with HAP at equal distances from PT and PR. Unless otherwise stated, SUs are randomly located around HAP within a circle of radius $\mathrm{10m}$. The results are averaged over $\mathrm{2000}$ channel realizations.

\begin{figure}
\centering
    \subfigure[]{\label{fig:rp_rs_basic}\includegraphics[scale=0.36]{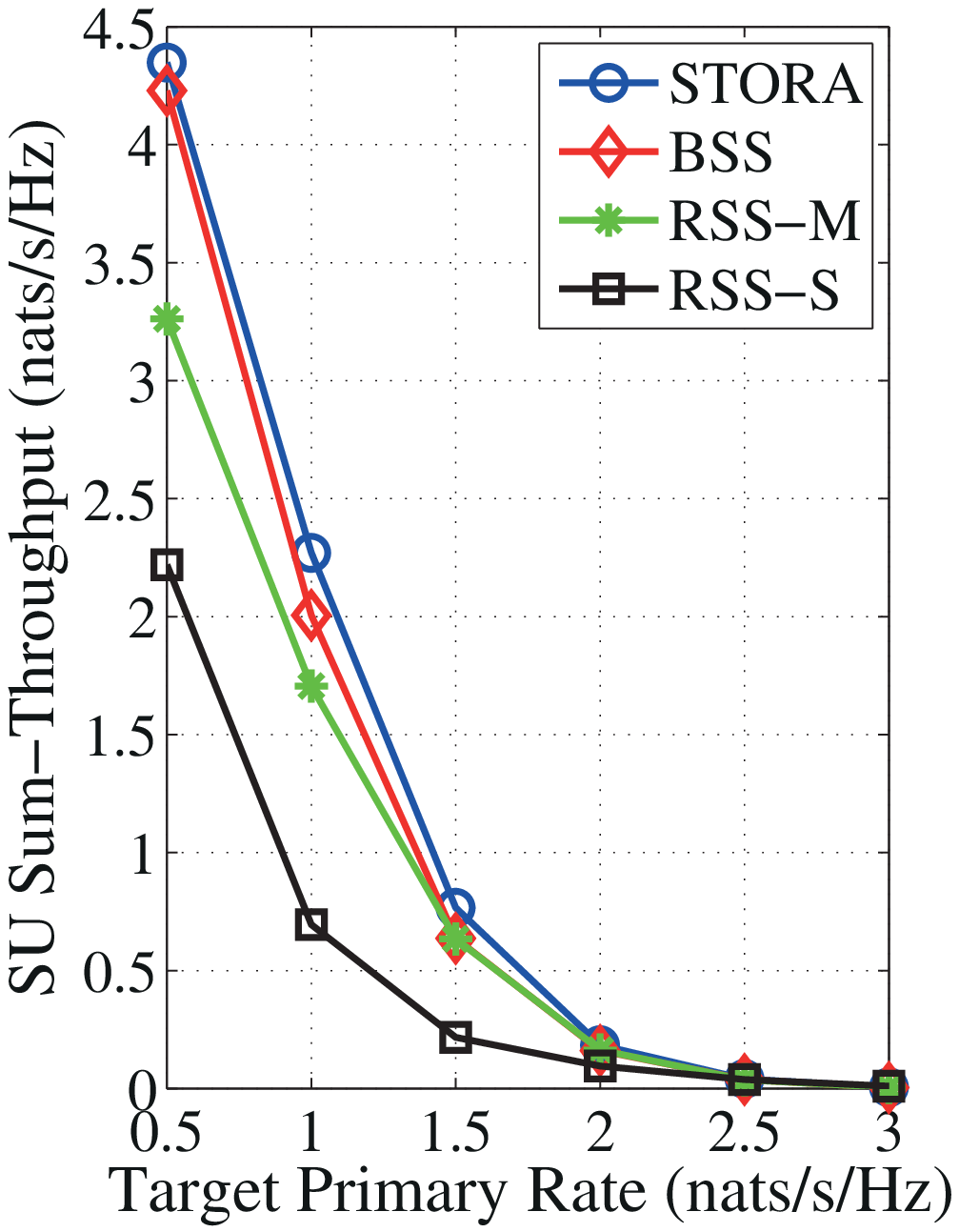}}
    \subfigure[]{\label{fig:rp_rs_coop}\includegraphics[scale=0.36]{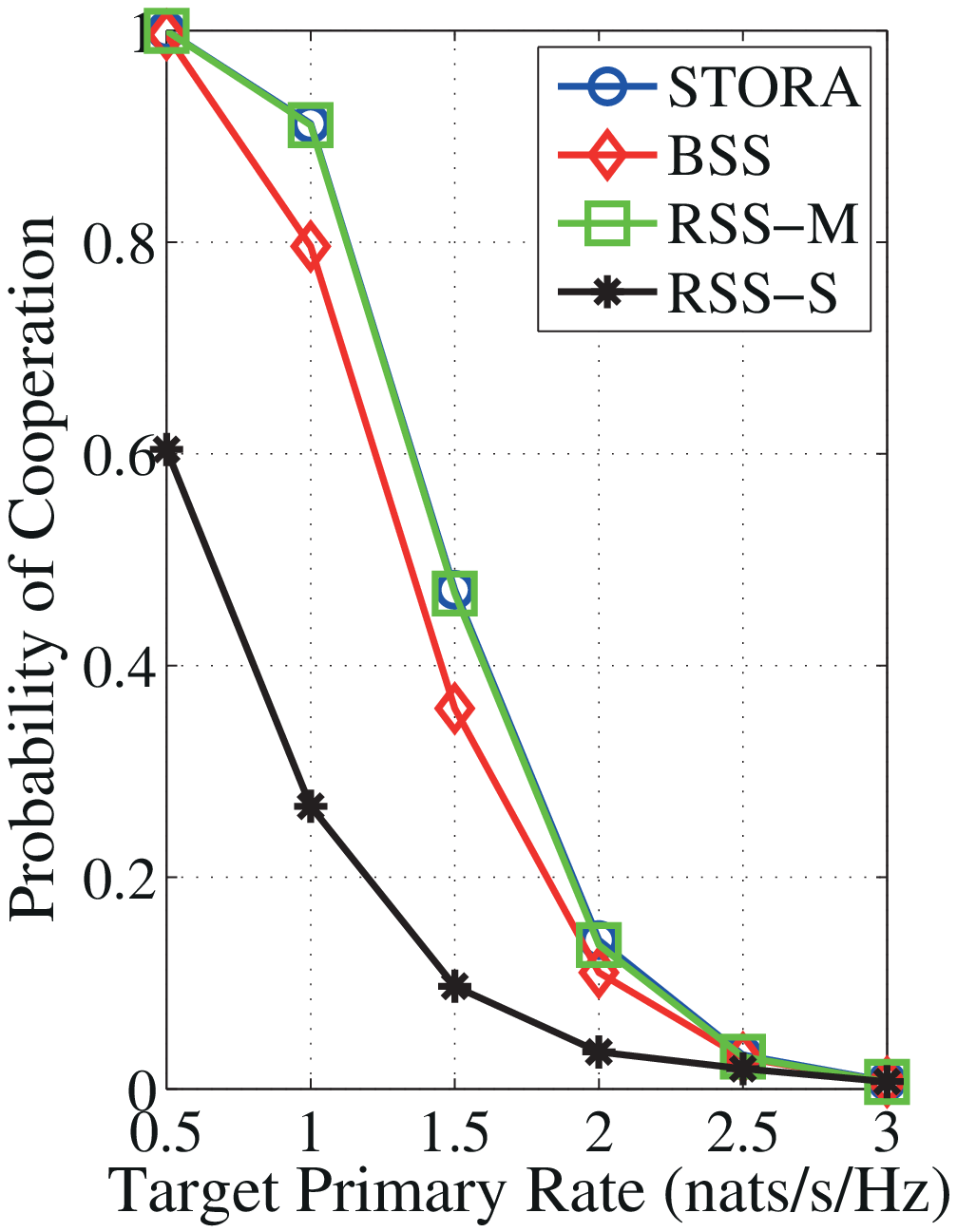}}
        \caption{Effect of target primary rate ($\bar{R}_{\mathrm{p}}$) on the (a) SU sum-throughput, (b) probability of cooperation. $N = \mathrm{4}$ and $P_{\mathrm{e}} = 20~\mathrm{dBm}$.}
      \label{fig:target}
    \end{figure}

\begin{figure*}
\centering
    \subfigure[SU sum-throughput $R_{\mathrm{s,sum}}$]{\label{fig:sum_thr}\includegraphics[scale=0.29]{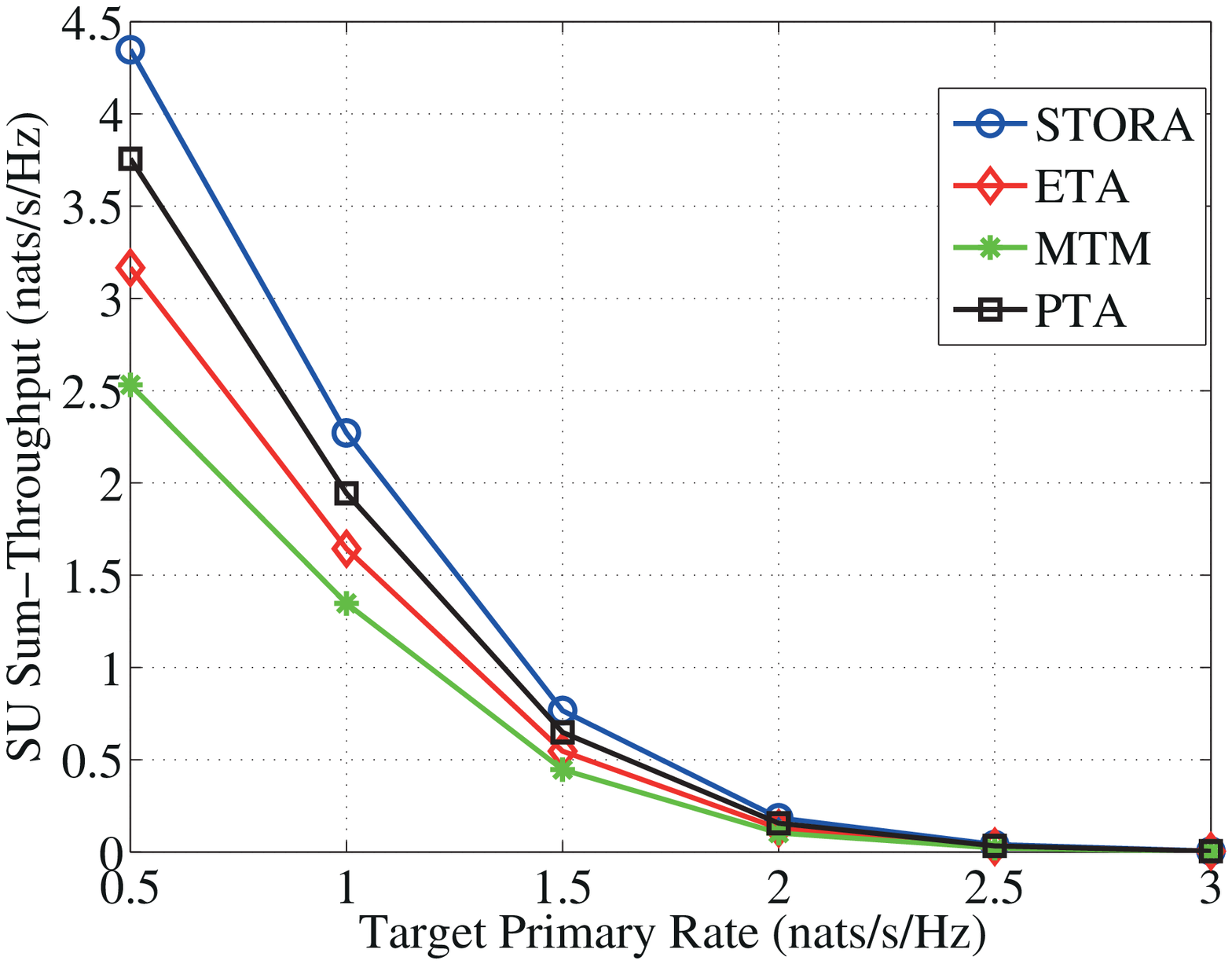}}
    \subfigure[Time spent by SUs on relaying $t_0$]{\label{fig:relaying_time}\includegraphics[scale=0.29]{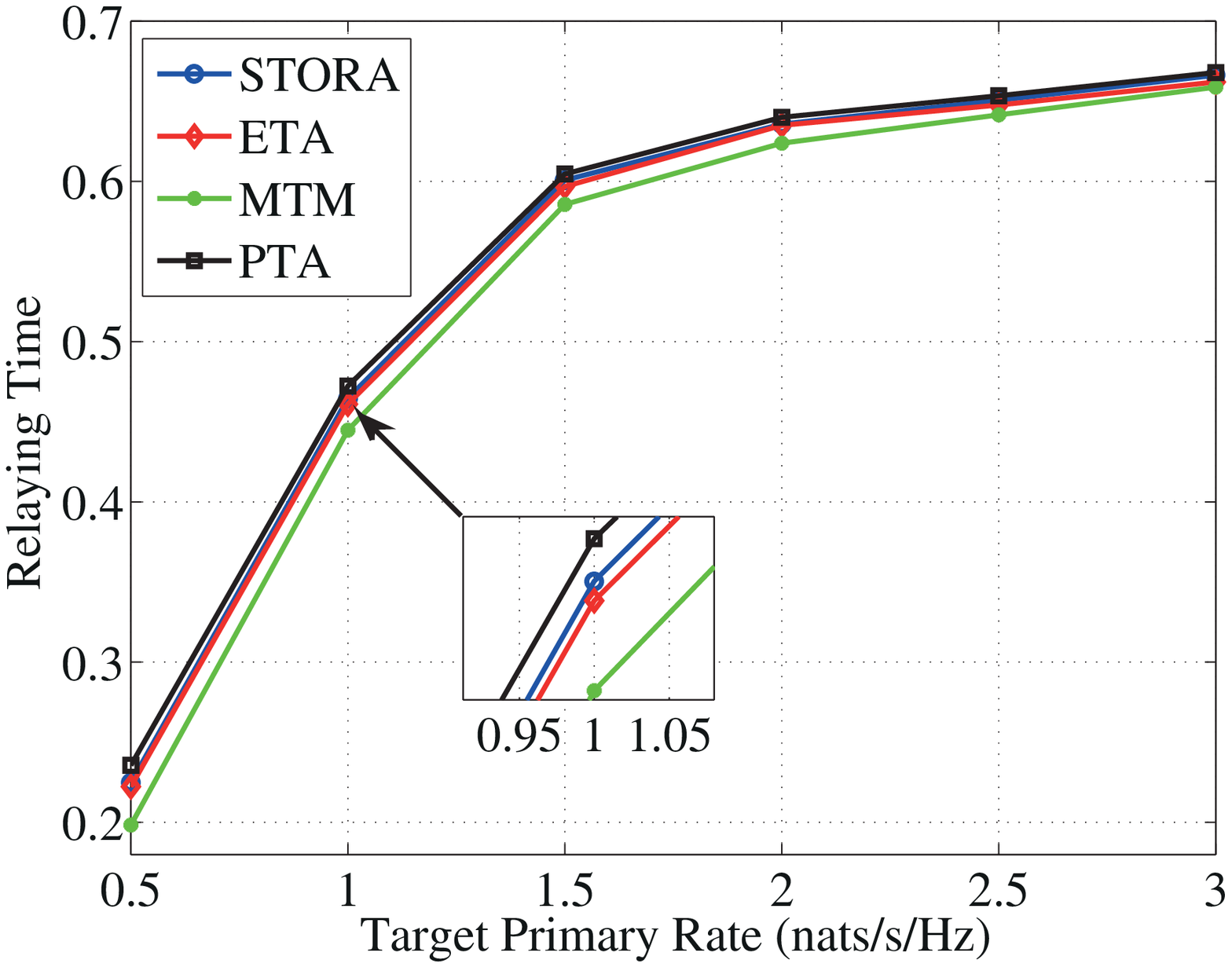}} 
        \subfigure[Rewarded time for SUs to access the spectrum $t_{\mathrm{a}}$]{\label{fig:rewarded_time}\includegraphics[scale=0.29]{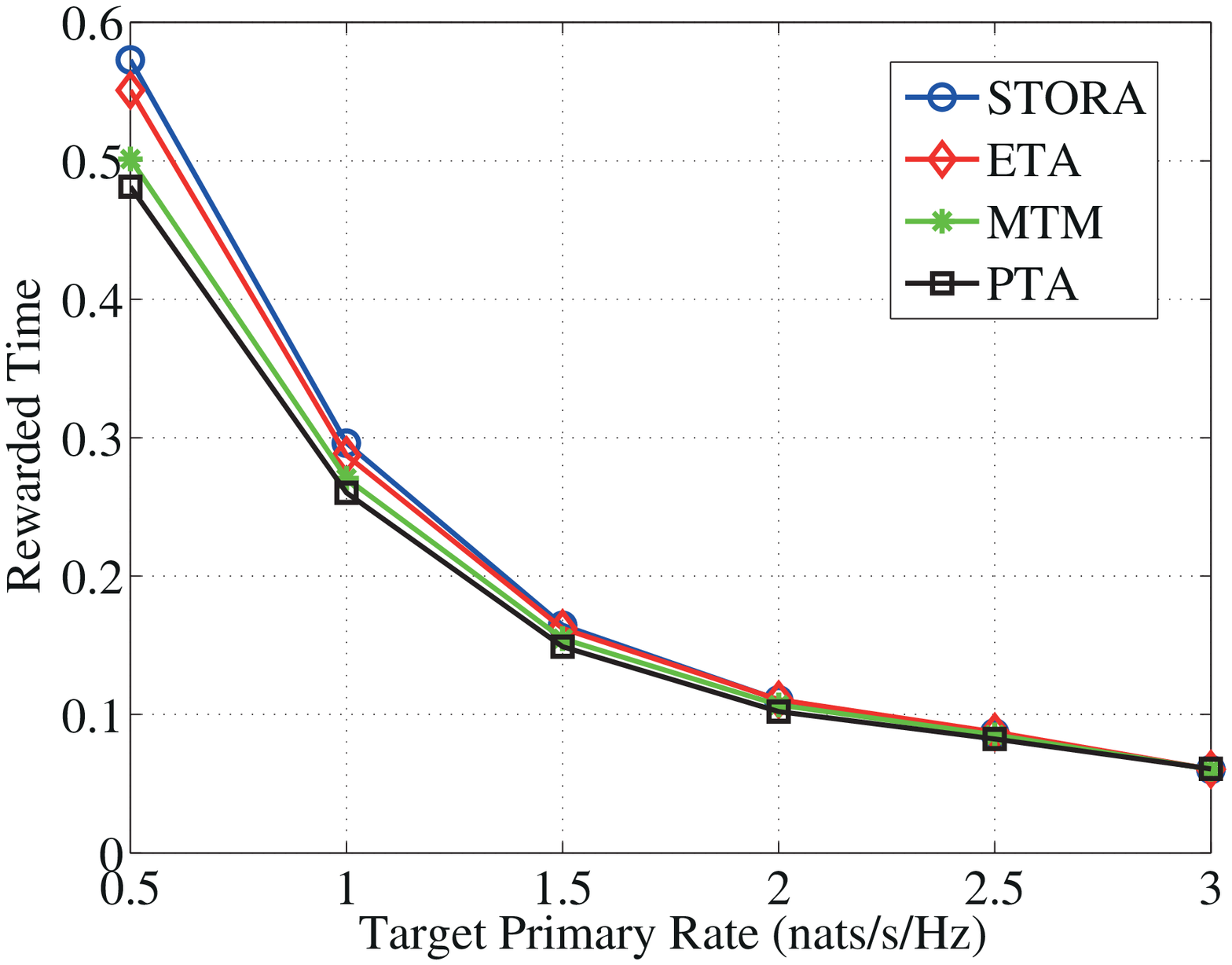}} 
        \caption{Effect of target primary rate ($\bar{R}_{\mathrm{p}}$) for different resource allocation schemes for $N = \mathrm{4}$ and $P_{\mathrm{e}} = \mathrm{20dBm}$.}
     \label{fig:rp_rs}
    \end{figure*}

\subsection{Performance Comparison of SU Selection Schemes (Fig.~\ref{fig:target})}
Fig.~\ref{fig:target} compares the performance of the STORA scheme with that of the best SU selection (BSS)~\cite{long}, the random single SU selection (RSS-S), and the random multiple SUs selection (RSS-M). The BSS allocates resources such that the SU sum-throughput is maximized with one SU relaying PU data; while in RSS-S, a random SU is chosen to relay PU data. In RSS-M, first a single SU is chosen randomly to relay PU data; if the cooperation from that chosen SU fails to meet the target primary rate, the second SU is chosen randomly for relaying, and so on.

Fig.~\ref{fig:rp_rs_basic} shows the effect of the target primary rate $\bar{R}_{\mathrm{p}}$ on SU sum-throughput.
In BSS scheme, a single SU may not be able to support $\bar{R}_{\mathrm{p}}$ due to insufficient harvested energy, failing to forge the cooperation (see Fig.~\ref{fig:rp_rs_coop}). This reduces the probability of cooperation and decreases the SU sum-throughput. The probability of cooperation is obtained as the ratio of number of channel realizations that result in successful cooperation between PU and SUs to the total number of channel realizations considered in simulations. In RSS-S scheme, the random selection of a single SU may not be able to meet $\bar{R}_{\mathrm{p}}$, which reduces the chances of cooperation. As Fig.~\ref{fig:rp_rs_coop} shows, RSS-M achieves the same probability of cooperation as that of STORA due to the selection of multiple SUs for relaying. But, since RSS-M chooses relaying SUs randomly, it fails to exploit the resources efficiently, resulting in the reduced SU sum-throughput. In Fig.~\ref{fig:rp_rs_coop}, the probability of cooperation less than one indicates that the primary rate constraint given by \eqref{eq:pqos} may not be satisfied through secondary cooperation in each fading state. 

\subsection{Effect of Primary Rate Constraint (Fig.~\ref{fig:rp_rs})}
As $\bar{R}_{\mathrm{p}}$ increases, the relaying SUs spend more time to relay PU data (see Fig.~\ref{fig:relaying_time}) which reduces SUs' spectrum access time (see Fig.~\ref{fig:rewarded_time}), deteriorating the SU sum-throughput. As shown in Fig.~\ref{fig:sum_thr}, the STORA scheme achieves the highest SU sum-throughput as it allocates resources efficiently. The ETA scheme allocates all SUs equal time for the spectrum access irrespective of their harvested energy, the residual energy left with them after relaying, and the channel gains.
As Fig.~\ref{fig:sum_thr} shows, MTM scheme's efforts to uplift SUs with unfavourable conditions to the level of SUs with favourable conditions expend the resources the least efficiently and achieve the lowest SU sum-throughput.
In PTA scheme,
the \textit{decode-relay-then-access} constraint, as Fig.~\ref{fig:relaying_time} shows, forces SUs to spend the highest time in decoding and relaying as SUs with bad links to PT consume additional time to decode PU data successfully, reducing SU sum-throughput compared to STORA scheme.

\begin{figure}
\centering
    \subfigure[]{\label{fig:users}\includegraphics[scale=0.36]{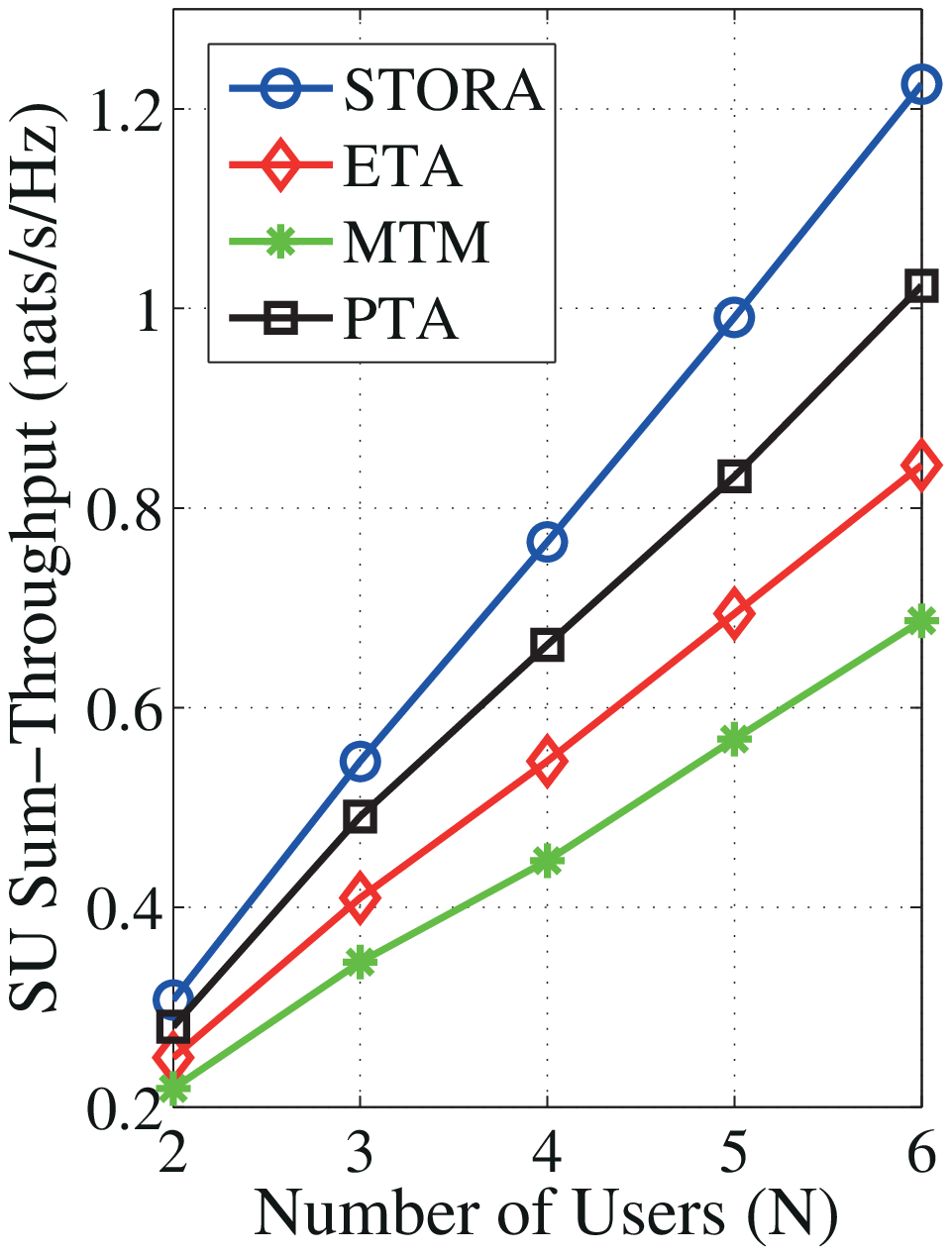}}
    \subfigure[]{\label{fig:fairness_users}\includegraphics[scale=0.36]{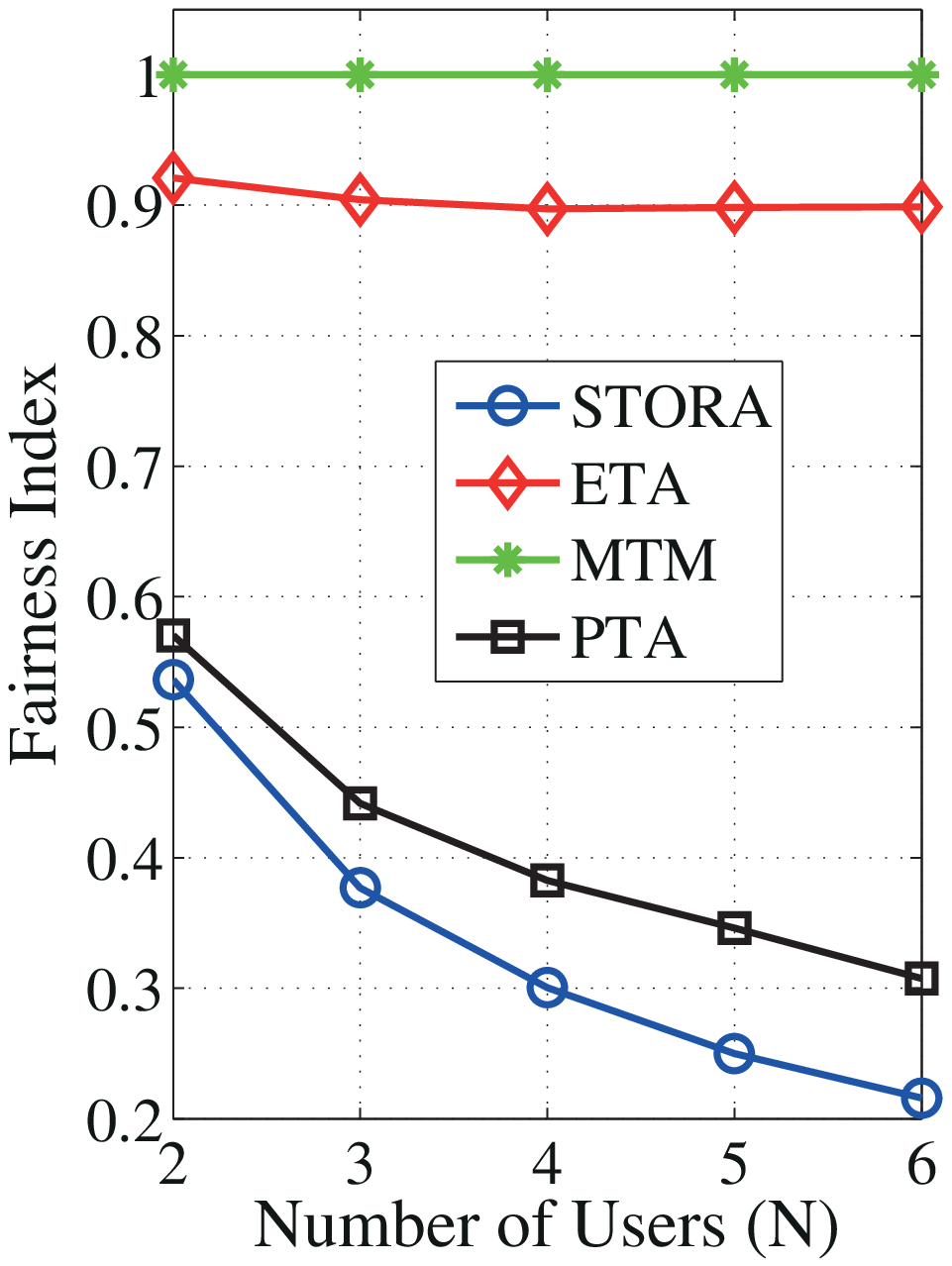}} 
        \caption{Effect of number of SUs on the (a) SU sum-throughput, (b) fairness. $\bar{R}_{\mathrm{p}} = \mathrm{1.5}$~$\mathrm{nats/s/Hz}$ and $P_{\mathrm{e}} = 20$ dBm.}
      \label{fig:users_thr_fai}
    \end{figure}
    
\subsection{Effect of Number of Secondary Users (Fig.~\ref{fig:users_thr_fai})}
The increase in number of SUs $N$ increases the user diversity as well as the possibility of having more users with higher harvested energy and good channel gains to HAP. This increases the number of SUs available to relay PU data, and time and energy available for the spectrum access. This improves SU sum-throughput.
Fig.~\ref{fig:fairness_users} shows the fairness achieved by each of the resource allocation schemes through Jain's fairness index $\mathcal{FI}$~\cite{jain} given by $\mathcal{FI} = \frac{\left(\sum_{i=1}^{N}x_i\right)^{2}}{N\sum_{i=1}^{N}x_i^2}$, where $x_i$ is the throughput of $i$th SU. Higher the fairness index, fairer is the scheme. Figs.~\ref{fig:users} and \ref{fig:fairness_users} together highlight the trade-off between SU sum-throughput and fairness. That is, STORA scheme achieves the highest sum-throughput, but is the least fairness inducing scheme; while MTM scheme is the fairest scheme as all SUs achieve the same throughput, but achieves the lowest sum-throughput as the resource allocation becomes inefficient from network's point of view. In ETA, though assigning equal time to each SU results in some fairness, it does not maximize the SU sum-throughput. In PTA, the contribution of an SU on the relaying link decides its time-share in the spectrum access. Though this provides better fairness compared to STORA, PTA scheme loses to STORA scheme in terms of the SU sum-throughput due to \textit{decode-relay-then-access} constraint.

\begin{figure}
\centering
    \subfigure[]{\label{fig:hap_thro}\includegraphics[scale=0.36]{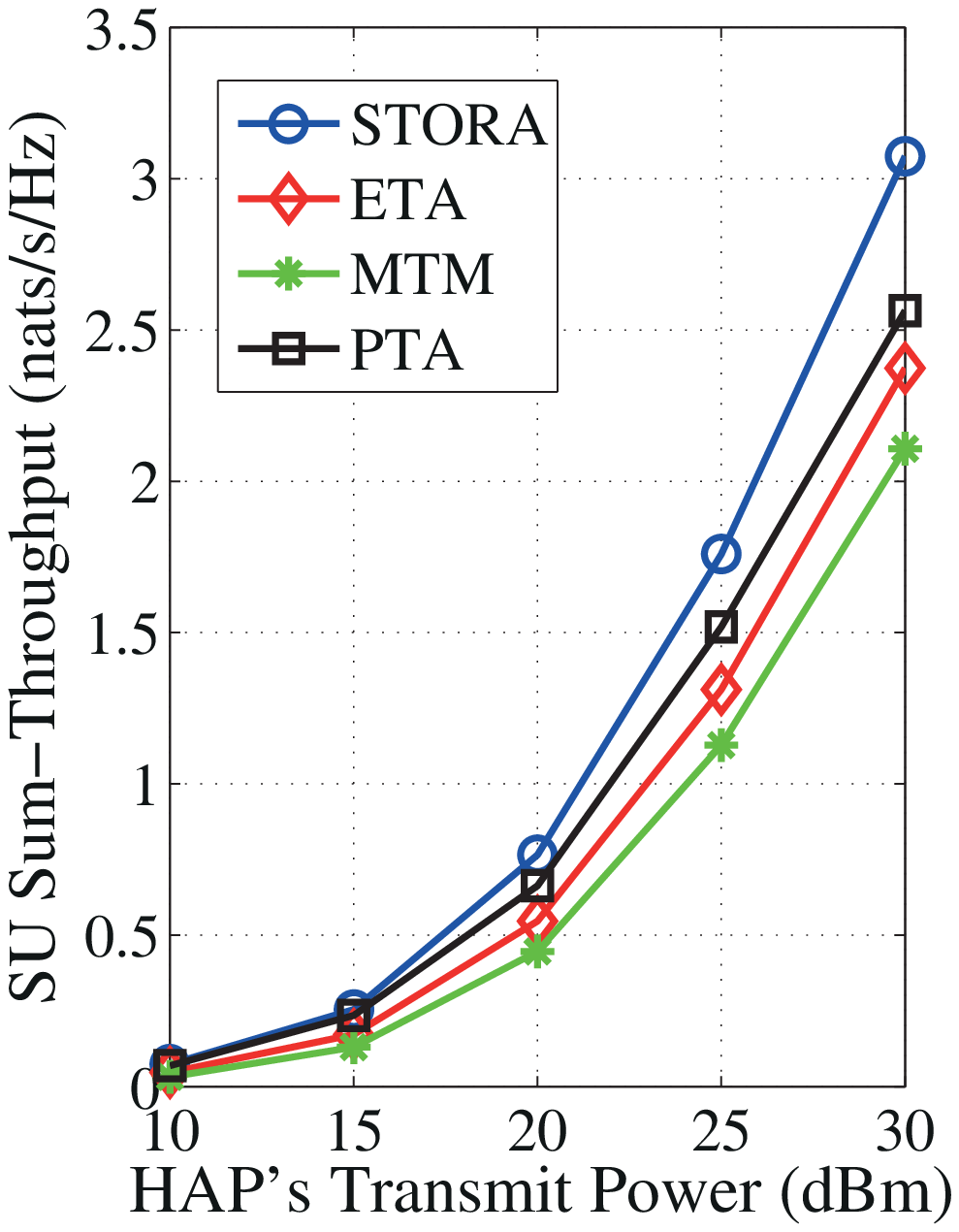}}
    \subfigure[]{\label{fig:hap_fairness}\includegraphics[scale=0.36]{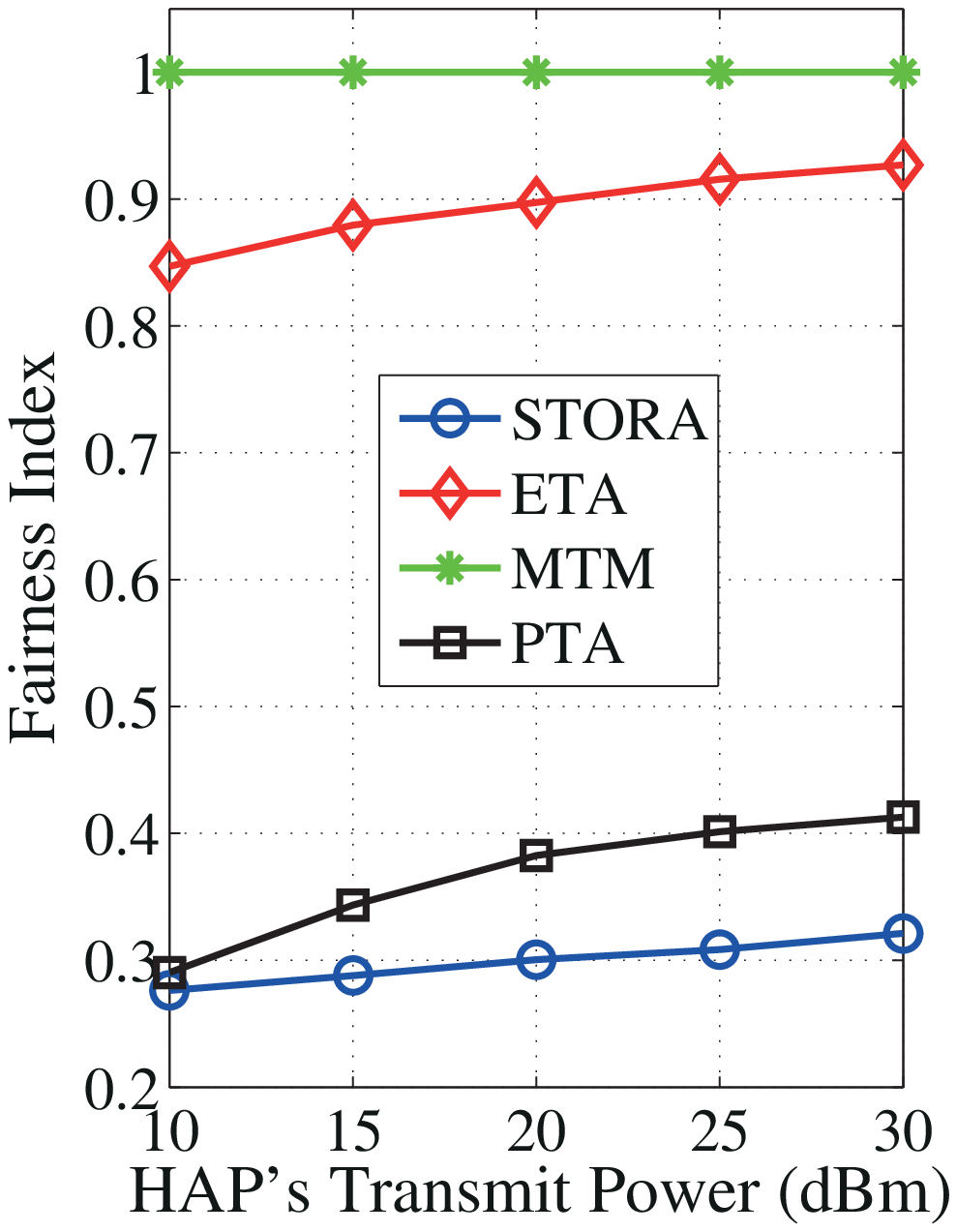}} 
        \caption{Effect of HAP's transmit power on (a) SU sum-throughput, (b) fairness. $\bar{R}_{\mathrm{p}} = \mathrm{1.5}$~$\mathrm{nats/s/Hz}$ and $N = 4$.}
      \label{fig:hap}
    \end{figure}

\subsection{Effect of HAP's Transmit Power (Fig.~\ref{fig:hap})}
Fig.~\ref{fig:hap_thro} depicts the effect of HAP's transmit power $P_{\mathrm{e}}$ on SU sum-throughput. As $P_{\mathrm{e}}$ increases, SUs harvest more energy from HAP's energy broadcast, which in turn increases the energy available for SUs' own transmissions for a given $\bar{R}_{\mathrm{p}}$. This improves SU sum-throughput.
Fig.~\ref{fig:hap_fairness} shows that the trend of fairness achieved by each of the four resource allocation schemes remains the same as that shown in Fig. \ref{fig:fairness_users}. 
\begin{figure}
\centering
\includegraphics[scale=0.37]{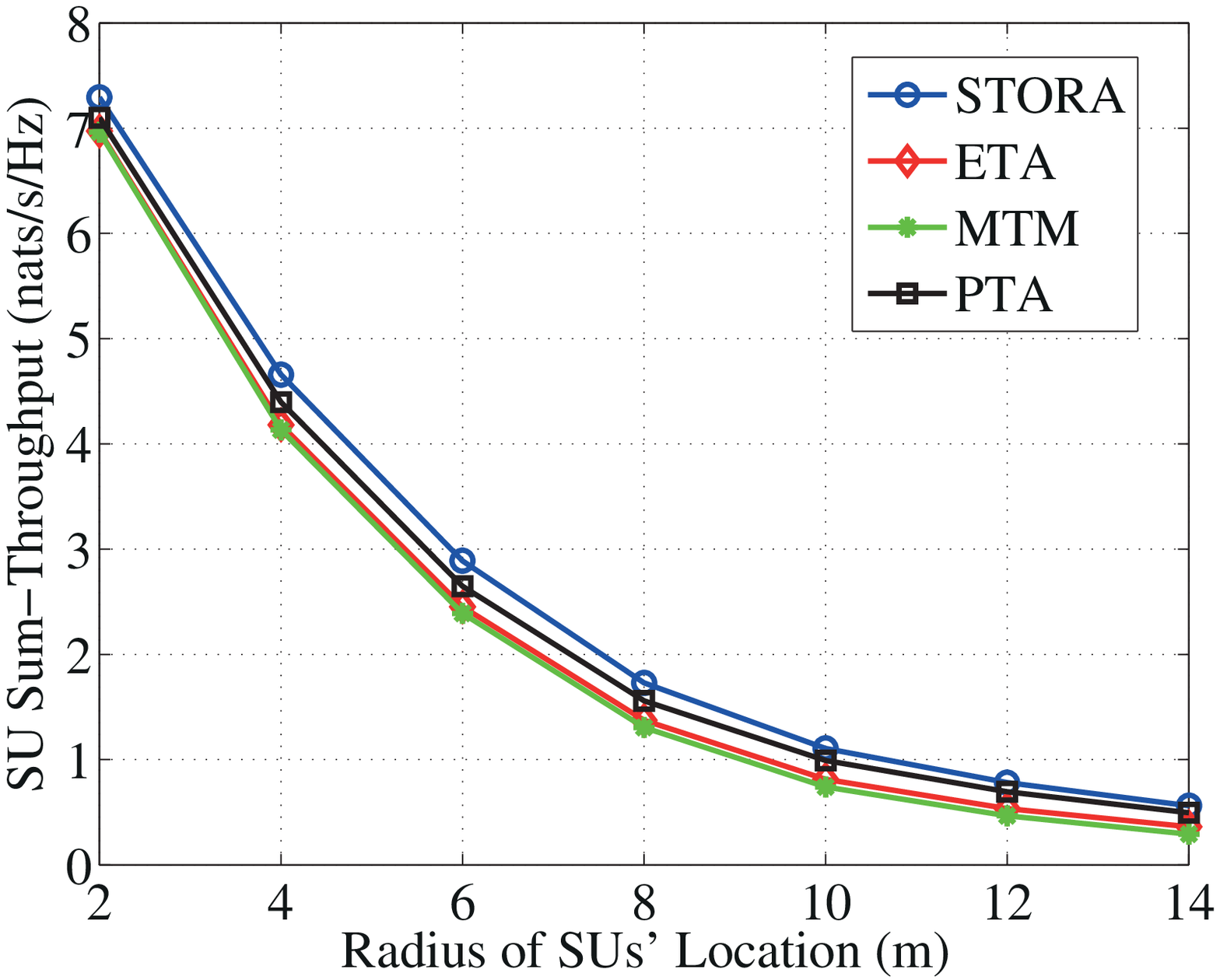}
\caption{Effect of SUs' distribution around HAP on SU sum-throughput for $N = \mathrm{4}$, $\bar{R}_{\mathrm{p}} = \mathrm{1.5}$~$\mathrm{nats/s/Hz}$, and $P_{\mathrm{e}} = 20$~$\mathrm{dBm}$.}
\label{fig:radius}
\end{figure}

\subsection{Effect of SUs' Location Radius Around HAP (Fig.~\ref{fig:radius})}
We consider that SUs are located randomly within a circle of radius $R$ meters around HAP, where the channels suffer from only path-loss attenuation. Fig.~\ref{fig:radius} shows that the increase in $R$ is unfavourable to the SU sum-throughput. The increase in $R$ increases the possibility that SUs are located farther from the HAP, which has doubly negative effect on SU sum-throughput. First, the harvested energy by SUs reduces due to path-loss, which ultimately reduces the energy available for relaying and SUs' own transmissions. As a result, the SU sum-throughput deteriorates. Second, in the rewarded time, SUs have to transmit data to HAP over longer distances, which causes the fall in the received SNR at HAP. This further declines SU sum-throughput. 

\section{Possible Variations to the Cooperation Protocol {and Future Directions}}
\label{sec:pos_var}
\subsection{Possible Variations to the Cooperation Protocol}
We now briefly discuss some possible variations to our proposed cooperation protocol given in Section~\ref{sec:coop_prot}.

1) In phase three, HAP may restart its energy broadcast and SUs that do not relay PU data harvest energy from it. But, the optimization problem, in this case, becomes quite involved. This is because, as we have shown in Section~\ref{sec:STORA}, the selection of relaying SUs, i.e., indirectly which SUs will harvest energy in phase three, depends on the decoding set, which in turn, depends on optimization variables $t_0$ and $t_{\mathrm{e}}$. Thus, the decoding set is unknown beforehand. On the other hand, the energy harvested in phase three will also affect the optimization variables $t_0$ and $t_{\mathrm{e}}$, which further impacts the decoding set and thus the relaying SU selection.

2) Instead of decode-and-forward relaying, SUs can use another popular relying protocol called \textit{amplify-and-forward} (AF) relaying. In AF relaying, SUs amplify the received data from PT and forward them to PR. Thus, SUs do not require to decode PU data in order to participate in relaying. Under this setup, SUs need not meet the decoding constraint given by \eqref{eq:dec_con} as it is no longer applicable to AF relaying.
{

\subsection{Future Directions}
We now discuss some interesting future directions that are worth investigating.

1) \textit{Imperfect CSI and synchronization}: The analysis in this paper is based on the assumption of perfect CSI and synchronization among all users. To characterize the effect of the uncertainty due to imperfect CSI and synchronization, we need robust problem formulation, which is an interesting future direction. Also, the effect of energy expenditure in CSI acquisition needs attention in energy harvesting systems.

2) \textit{Energy efficiency}: In our work, we have two energy sources for SUs: i) HAP's energy broadcast, ii) PU transmissions. HAP's broadcast improves the harvested energy significantly, in turn enhancing QoS. But, this consumes extra energy from the grid, reducing energy efficiency (EE) compared to the case of harvesting energy from PU transmissions only. This opens up an interesting future direction to study the trade-off between QoS and EE.}

\section{Concluding Remarks}
\label{sec:conclusion}
The integration of a wireless powered communication network with a cooperative cognitive radio network allows us to reap the advantages of the both together. Given the primary rate constraint, STORA scheme jointly performs the relay selection and energy and time allocation for SUs. But, STORA scheme creates unfairness among SUs as some SUs may get zero individual throughput. Other investigated resource allocation schemes like ETA, MTM, and PTA enhance fairness, but achieve lower secondary sum-throughput compared to STORA scheme due to the throughput-fairness trade-off. By allowing SUs to obtain equal throughput under all conditions, MTM scheme becomes the fairest scheme among the four schemes, but sacrifices the most on sum-throughput. Though the \textit{decode-relay-then-access} constraint puts PTA scheme behind STORA scheme in terms of the SU sum-throughput, PTA still seems to be a good compromise between the sum-throughput and fairness. The STORA scheme provides a mechanism to trade-off fairness with overall throughput. In fact, STORA scheme generalizes the existing algorithms, because the optimal algorithm proposed for STORA can be modified to accommodate different fairness constraints considered in the paper.

\appendices
\section{Proof of Proposition~\ref{lemma1}}
\label{app:opt_soln_STORA}
The optimal primal and dual variables must satisfy the KKT stationarity conditions.
Taking derivatives  of $\mathcal{L}_{1}$ with respect to $(E_{i\mathrm{h}}$, $E_{i\mathrm{p}}) $ and $\mathcal{L}_{2}$ with respect to $(t_{0}, t_{i})$, we obtain
\begin{subequations}
\begin{align}
& \frac{\gamma_{i\mathrm{h}}}{1+x^{*}_{i}} - \mu^{*}_{i} =  0, \hspace{8mm} \forall  i, \label{eq:kkt1} \\
&   \frac{\lambda \gamma_{i\mathrm{p}}}{1+\gamma_{\mathrm{p}} + y^{*}} - \mu^{*}_{i}  =  0, \hspace{1mm} \forall  i, \label{eq:kkt2}\\
 & \lambda^{*}\ln\left(1+ \gamma_{\mathrm{p}} + y^{*} \right) - \frac{\lambda^{*} y^{*}}{1 + \gamma_{\mathrm{p}} +y^{*}} + \kappa^* Q_2 - 2\nu^* =  0,  \label{eq:kkt4} \\
 & \ln\left( 1+x^{*}_{i} \right) - \frac{x^{*}_{i}}{1+x^{*}_{i}} - \nu^* =  0, \hspace{8mm} \forall i, \label{eq:kkt3}
\end{align}
\end{subequations}
where $ x^{*}_{i} = \frac{\gamma_{i\mathrm{h}}E^{*}_{i\mathrm{h}}}{t^{*}_{i}}$, $y^{*} = \frac{\sum_{\mathrm{SU}_i\in \mathcal{S_D} }\gamma_{i\mathrm{p}}E^{*}_{i\mathrm{p}} }{t^{*}_{0}}$. The corresponding KKT complementary slackness conditions are
\begin{subequations}
\begin{align}
& \lambda^{*}\left(\bar{R}_{\mathrm{p}} - t^{*}_{0}\ln\left(1+ \gamma_{\mathrm{p}} + \frac{\sum_{\mathrm{SU}_i \in \mathcal{S_D}} \gamma_{i\mathrm{p}}E^{*}_{i\mathrm{p}} }{t^{*}_{0}}\right)\right) = 0, \label{eq:kkt6}\\
& \kappa^* \left( \bar{R}_{\mathrm{p}} - Q_1 t^*_{\mathrm{e}} - Q_2 t^*_{0}  \right)  = 0,\label{eq:kkt7} \\
& \mu^{*}_{i}(E^{*}_{i\mathrm{p}} + E^{*}_{i\mathrm{h}} - (P_{\mathrm{e}} + \theta_i )t^{*}_{\mathrm{e}}) = 0, \hspace{1mm} \forall  i, \label{eq:kkt8}\\
& \nu^{*} ( t^{*}_{\mathrm{e}} +2t^{*}_{0}+\sum_{i=1}^{N}t^{*}_{i}-1)  = 0. \label{eq:kkt9}
\end{align}
\end{subequations}
We learn from Proposition \ref{prop1} that the primary rate constraint \eqref{eq:qos}, energy neutrality constraint \eqref{eq:conv1}, and total time constraint \eqref{eq:conv2} must be satisfied with equality. Therefore, we infer from \eqref{eq:kkt6}-\eqref{eq:kkt9} that $(\lambda, \boldsymbol{\mu}, \nu) > 0$ and $\kappa \geq 0$. On rearranging the  equations \eqref{eq:kkt1}-\eqref{eq:kkt3} and using \eqref{eq:kkt9}, we write the optimal solution for the STORA problem as given in Proposition \ref{lemma1}.

\section{Proof of Proposition~\ref{eta_teq}}
\label{app:prop_eta}
The Lagrangian $\mathcal{L}_{3}$ of the time allocation subproblem of ETA is given by
\begin{align}
 \mathcal{L}_{3} &= \sum_{i=1}^N R_{i}(E_{i\mathrm{h}},t_{\mathrm{eq}} ) - \lambda \left(\bar{R}_{\mathrm{p}} - R_{\mathrm{p,c}}(\boldsymbol{E_{\mathrm{sp}}} , t_{\mathrm{e}}, t_{0})\right)  \nonumber \\
 & - \kappa \left( \bar{R}_{\mathrm{p}} - Q_1 t_{\mathrm{e}} - Q_2 t_{0} \right)-  \nu (2t_{0}+t_{\mathrm{e}}+\sum_{i=1}^{N}t_{i}-1).
 \label{eq:eta_lag}
\end{align}
The optimal $t^{*}_{\mathrm{eq}}$ that maximizes the Lagrangian $\mathcal{L}_3$ can be obtained from
\begin{equation}
\frac{\partial \mathcal{L}_{3}}{\partial t_{\mathrm{eq}}} = \sum_{i=1}^{N} \ln\left(1+ \frac{\gamma_{i\mathrm{h}}E^{*}_{i\mathrm{h}} }{t^{*}_{\mathrm{eq}}} \right) - \frac{\frac{\gamma_{i\mathrm{h}}E^{*}_{i\mathrm{h}} }{t^{*}_{\mathrm{eq}}}  }{1+\frac{\gamma_{i\mathrm{h}}E^{*}_{i\mathrm{h}} }{t^{*}_{\mathrm{eq}}}} - N\nu^{*} = 0,
\end{equation}
which forms Proposition \ref{eta_teq}.
\section{Proof of Proposition~\ref{opt_soln_mmf}}
\label{app_prop_mmf}
The energy and time allocation on the access link that maximize \eqref{eq:dual1_mmf} and \eqref{eq:dual2_mmf} are found using KKT stationarity conditions. On differentiating \eqref{eq:dual1_mmf} and \eqref{eq:dual2_mmf} with respect to $E_{i\mathrm{h}}$ and $t_{i}$, we obtain 
\begin{align}
& \rho_{i}\frac{\gamma_{i\mathrm{h}}}{1+\frac{\gamma_{i\mathrm{h}}E^{*}_{i\mathrm{h}}}{t^{*}_{i}}} - \mu^{*}_{i} =  0, \hspace{8mm} \forall  i \label{eq:mmf1} \\
 & \ln\left( 1+\frac{\gamma_{i\mathrm{h}}E^{*}_{i\mathrm{h}}}{t^{*}_{i}} \right) - \frac{\frac{\gamma_{i\mathrm{h}}E^{*}_{i\mathrm{h}}}{t^{*}_{i}}}{1+\frac{\gamma_{i\mathrm{h}}E^{*}_{i\mathrm{h}}}{t^{*}_{i}}} - \frac{\nu^*}{\rho^*_i} =  0, \hspace{5mm} \forall  i. \label{eq:mmf2}
\end{align}
The KKT stationarity conditions with respect to $E_{i\mathrm{p}}$ and $t_{0}$ are same as that of STORA as given in \eqref{eq:kkt2} and \eqref{eq:kkt4}. The primal variable $t_{\mathrm{e}}$ is obtained using KKT complementary slackness condition \eqref{eq:kkt9} and $R_{\min}$ is found using the iterative algorithm discussed in section \ref{sec:mmf}. On rearranging~\eqref{eq:mmf1} and \eqref{eq:mmf2}, we obtain the results in Proposition \ref{opt_soln_mmf}.

\section{Proof of Proposition~\ref{opt_soln_prop}}
\label{app_prop_pf}
The Lagrangian $\mathcal{L}_{5}$ of the energy allocation subproblem for the PTA scheme is
\begin{align*}
\mathcal{L}_{5} & = \sum_{i=1}^N R_{i}(E_{i\mathrm{h}}, \zeta\gamma_{i\mathrm{p}}E_{i\mathrm{p}}) - \lambda \left(\bar{R}_{\mathrm{p}} - R_{\mathrm{p,c}}(\boldsymbol{E_{\mathrm{sp}}} , t_{\mathrm{e}}, t_{0})\right) \nonumber \\
& - \sum_{i = 1}^{N} \mu_{i}(E_{i\mathrm{p}} + E_{i\mathrm{h}} - (P_{\mathrm{e}} + \theta_i)t_{\mathrm{e}}).
\end{align*}
The first order stationarity condition of $\mathcal{L}_{5}$ with respect to $E_{i\mathrm{h}}$ and $E_{i\mathrm{p}}$ are  given by
\begin{align}
&\frac{\partial \mathcal{L}_{5}}{\partial E_{i\mathrm{h}}} =  \frac{\gamma_{i\mathrm{h}}}{1+\frac{\gamma_{i\mathrm{h}}E^{*}_{i\mathrm{h}}}{\zeta\gamma_{i\mathrm{p}}E^{*}_{i\mathrm{p}}}} - \mu^{*}_{i} =  0, \nonumber \\
&\frac{\partial \mathcal{L}_{5}}{\partial E_{i\mathrm{p}}} =  \zeta^{*} \gamma_{i\mathrm{p}} \ln\left(1+{z^{*}_{i}}\right) - \frac{z^{*}_{i}}{1+z^{*}_{i}}, \nonumber \\
&- \frac{\lambda^{*} \gamma_{i\mathrm{p}}}{1+\gamma_{\mathrm{p}}+\frac{\sum_{i=1}^{N}\gamma_{i\mathrm{p}}E^{*}_{i\mathrm{p}}}{t^{*}_{0}}} - {\mu^{*}_{i}}- {\zeta^{*} \nu^{*}} = 0,
\end{align}
with $z^{*}_{i} = \frac{\gamma_{i\mathrm{h}}E^{*}_{i\mathrm{h}}}{\zeta^{*}\gamma_{i\mathrm{p}}E^{*}_{i\mathrm{p}}}$. The solution of $t_{0}$ and $t_{\mathrm{e}}$ is same as that of STORA scheme. The solution of the above equations and the feasibility condition $t^{*}_{i} = \gamma_{i\mathrm{p}} \zeta^{*} E^{*}_{i\mathrm{p}}$ $\forall i$ form Proposition \ref{opt_soln_prop}.

\bibliographystyle{ieeetr}
\bibliography{paper}

\end{document}